\documentclass[a4paper,11pt,oneside]{article}
\usepackage{latexsym}
\usepackage{fancyhdr}
\usepackage[mathscr]{eucal}
\usepackage{amsmath}
\usepackage{mathrsfs}
\usepackage{amsthm}
\usepackage{amsfonts} 
\usepackage{amssymb} 
\usepackage{amscd}
\usepackage{bbm}
\usepackage{graphicx}
\usepackage{graphics}
\usepackage[left=2.5cm,right=2.5cm,top=3cm,bottom=2cm]{geometry}

\usepackage{bm}

\newcommand{\ii}{\mathrm{i}}
\newcommand{\cH}{\mathcal{H}}
\newcommand{\ran}{\mathrm{ran}}
\newcommand{\ud}{\mathrm{d}}
 \newtheorem{thm}{Theorem}[section]
 
 \newtheorem{lem}[thm]{Lemma}
 \newtheorem{prop}[thm]{Proposition}
 \theoremstyle{definition}
 
 \theoremstyle{remark}
 \newtheorem{rem}[thm]{Remark}
 
 \numberwithin{equation}{section}

\begin{document}\hyphenation{Cou-lomb}

%
%
%
%
%
%
%
%
%

\title{Self-adjoint realisations of the Dirac-Coulomb Hamiltonian for heavy nuclei\footnote{This work was partially supported by the 2014-2017 MIUR-FIR grant ``\emph{Cond-Math: Condensed Matter and Mathematical Physics}'' code RBFR13WAET.}}

\author{Matteo Gallone\footnote{International School for Advanced Studies -- SISSA, via Bonomea 265, 34136 Trieste, Italy. e-mail 
\texttt{mgallone@sissa.it}} $\,$ and Alessandro Michelangeli\footnote{International School for Advanced Studies -- SISSA, via Bonomea 265 34136 Trieste, Italy.
e-mail 	\texttt{alemiche@sissa.it}}}


%
%
%
%

\date{}
\maketitle
\begin{abstract}
We derive a classification of the self-adjoint extensions of the three-dimensional Dirac-Coulomb operator in the critical regime of the Coulomb coupling. Our approach is solely based upon the Kre{\u\i}n-Vi\v{s}ik-Birman extension scheme, or also on Grubb's universal classification theory, as opposite to previous works within the standard von Neumann framework. This let the boundary condition of self-adjointness emerge, neatly and intrinsically, as a multiplicative constraint between regular and singular part of the functions in the domain of the extension, the multiplicative constant giving also immediate information on the invertibility property and on the resolvent and spectral gap of the extension. 

\textbf{Mathematics Subject Classification (2010).} 47B25; 47N20; 47N50; 81Q10.

\textbf{Keywords.} Dirac-Coulomb operator, self-adjoint extensions, Kre{\u\i}n-Vi\v{s}ik-Birman extension theory, Grubb's universal classification
\end{abstract}


\section{Introduction}

In quantum mechanics a relativistic electron or positron (or more generally a relativistic spin-$\frac{1}{2}$ particle) which moves freely in the three-dimensional space is described by elements of the Hilbert space
\begin{equation}
 \cH\;:=\;L^2(\mathbb{R}^3)\otimes\mathbb{C}^4\;\cong\;L^2(\mathbb{R}^3,\mathbb{C}^4,\ud x)
\end{equation}
and by the (formal) Hamiltonian
\begin{equation}
 H_0\;:=\;-\ii c \hbar \,\bm{\alpha}\cdot \bm{\nabla}+\beta m c^2
\end{equation}
acting on $\cH$, where $\hbar$ is Planck's constant, $c$ is the speed of light, $m$ is the mass of the particle, and $\bm{\alpha}\equiv(\alpha_1,\alpha_2,\alpha_3)$ and $\beta$ are the $4\times 4$ matrices
\begin{equation}
 \beta\;=\;\begin{pmatrix} 
            \mathbbm{1} & \mathbbm{O} \\
            \mathbbm{O} & -\mathbbm{1}
           \end{pmatrix}\,,\qquad
 \alpha_j\;=\;\begin{pmatrix}
               \mathbbm{O} & \sigma_j \\
               \sigma_j & \mathbbm{O}
              \end{pmatrix}\,,\qquad j\in\{1,2,3\}\,,
\end{equation}
having denoted by $\mathbbm{1}$ and $\mathbbm{O}$, respectively, the identity and the zero $2\times 2$ matrix, and by $\sigma_j$, as customary, the Pauli matrices
\begin{equation}
 \sigma_1\;=\;\begin{pmatrix}
               0 & 1 \\ 1 & 0
              \end{pmatrix}\,,\qquad
 \sigma_2\;=\;\begin{pmatrix}
               0 & -\ii \\ \ii & 0
              \end{pmatrix}\,,\qquad
 \sigma_3\;=\;\begin{pmatrix}
               1 & 0 \\ 0 & -1
              \end{pmatrix}\,.
\end{equation}
Explicitly, the scalar product between any two elements $\psi\equiv(\psi_1,\psi_2,\psi_3,\psi_4)$ and $\phi\equiv(\phi_1,\phi_2,\phi_3,\phi_4)$ in $\cH$ is given by
\begin{equation}
 \langle\psi,\phi\rangle_\cH\;=\;\sum_{j=1}^4\int_{\mathbb{R}^3}\overline{\psi_j(x)}\,\phi_j(x)\,\ud x\,,
\end{equation}
and $H_0$ is the first order matrix-valued differential operator
\begin{equation}
 H_0\;=\;\begin{pmatrix}
          m c^2\mathbbm{1} & -\ii\hbar c \,\bm{\sigma}\cdot\bm{\nabla} \\
          -\ii\hbar c \,\bm{\sigma}\cdot\bm{\nabla} &  -mc^2\mathbbm{1}
         \end{pmatrix}
\end{equation}
(where $\bm{\sigma}\equiv(\sigma_1,\sigma_2,\sigma_3)$), known as the free Dirac operator.

The properties of $H_0$ are well known \cite{Thaller-Dirac-1992}. $H_0$ is essentially self-adjoint on $C^\infty_0(\mathbb{R}^3\!\setminus\!\{0\})\otimes\mathbb{C}^4$ with domain of self-adjointness
\begin{equation}
 H^1(\mathbb{R}^3)\otimes\mathbb{C}^4\;\cong\; H^1(\mathbb{R}^3,\mathbb{C}^4)\,,
\end{equation}
 and its spectrum (as a self-adjoint operator on $\cH$) is purely absolutely continuous and given by
\begin{equation}
 \sigma(H_0)\;=\;\sigma_{\mathrm{ac}}(H_0)\;=\;(-\infty,-mc^2]\cup[mc^2,+\infty)\,.
\end{equation}
In fact, $H_0$ is unitarily equivalent to
\begin{equation}
 \widetilde{H_0}\;:=\;\begin{pmatrix}
                       \mathbbm{1}\sqrt{-c^2\Delta+m^2c^4} & 0 \\
                       0 & -\mathbbm{1}\sqrt{-c^2\Delta+m^2c^4}
                      \end{pmatrix}.
\end{equation}

When the particle is subject to the external scalar field due to the Coulomb interaction with a nucleus of atomic number $Z$ placed in the origin of $\mathbb{R}^3$, this is accounted for by the so-called Dirac-Coulomb Hamiltonian
\begin{equation}\label{eq:Hformal}
 H\;:=\;-\ii c \hbar \,\bm{\alpha}\cdot \bm{\nabla}+\beta m c^2-\frac{e^2 Z}{\hbar\,|x|}\mathbbm{1}\;=\;H_0-\frac{\,cZ\alpha_\mathrm{f}\,}{|x|}\mathbbm{1}\,,
\end{equation}
where now $\mathbbm{1}$ is the $4\times 4$ identity matrix (no confusion should arise here and henceforth on the symbol $\mathbbm{1}$, being its meaning of identity self-explanatory from the context), $e$ is the elementary charge, and 
\begin{equation}
 \alpha_\mathrm{f}\;=\;\frac{e^2}{\hbar c}\;\approx\;\frac{1}{137}
\end{equation}
is the fine-structure constant. The operator $H$ can at least be defined minimally on $C^\infty_0(\mathbb{R}^3\!\setminus\!\{0\})\otimes\mathbb{C}^4$, in which case it is densely defined and symmetric on $\cH$.
However, the possibility that this yields an unambiguous physical realisation of $H$ depends on the magnitude of the coupling $Z\alpha_\mathrm{f}$, hence of the nuclear charge $Z$. It is indeed well known \cite{Thaller-Dirac-1992} that the formal operator \eqref{eq:Hformal} is essentially self-adjoint on $C^\infty_0(\mathbb{R}^3\!\setminus\!\{0\},\mathbb{C}^4)$ \emph{only} when $Z\alpha_\mathrm{f}\leqslant\frac{\sqrt{3}}{\,2}$ (i.e., $Z\leqslant 118$), in which case the domain of self-adjointness is $\mathcal{D}(H)=\mathcal{D}(H_0)=H^1(\mathbb{R}^3,\mathbb{C}^4)$ and the spectrum consists of the same essential part $\sigma_\mathrm{ess}(H)=(-\infty,-mc^2]\cup[mc^2,+\infty)$ as for $H_0$, plus a discrete spectrum in the `gap' $(-mc^2,mc^2)$ consisting of eigenvalues $E_{n,\kappa}$ given by Sommerfeld's celebrated fine-structure formula
\begin{equation}
E_{n,\kappa}\;=\;mc^2\Big(1+\frac{(Z\alpha_\mathrm{f})^2}{\big(n+\sqrt{\kappa^2-(Z\alpha_\mathrm{f})^2}\,\big)^{\!2}}\Big)^{\!-\frac{1}{2}},\quad n\in\mathbb{N}_0,\,\kappa\in\mathbb{Z}\!\setminus\!\{0\}\,.
\end{equation}

Although the above regime of $Z$ covers all currently known elements (the last one to be discovered, the Oganesson ${}^{294}_{118}$Og, thus $Z=118$, was first synthesized in 2002 and formally named in 2016), the problem of the self-adjoint realisation of the Dirac-Coulomb Hamiltonian above the threshold $Z\alpha_\mathrm{f}=\frac{\sqrt{3}}{\,2}$ has been topical since long and so is still today. Even the consideration that the problem only arises due to the idealisation of point-like nuclei (and also because one neglects the anomalous magnetic moment of the electron) does not diminish its relevance, given the extreme experimental precision, for example, of Sommerfeld's fine-structure formula for the eigenvalues of $H$ when $Z\leqslant 118$.

From the mathematical side, the study of the self-adjoint extensions of the Dirac-Coulomb Hamiltonian has a long and active history \cite{Evans-1970,Weidmann-1971,Rejto-1971,Schmincke-1972,Schmincke-1972-distinguished,Gustafson-Rejto-1973,Wust-1975,Kalf-Schmincke-Wust-1975,Nenciu-1976,Wust-1977,Chernoff-1977,Klaus-Wust-1978,Landgren-Rejto-JMP1979,Landgren-Rejto-Martin-JMP1980,Landgren-Rejto-1981,Burnap-Brysk-Zweifel-NuovoCimento1981,Arai-Yamada-RIMS-1982,Arai-Yamada-RIMS-1983,Kato-1983,Berthier-Georgescu-JFA1987,Thaller-Dirac-1992,Xia-1999,Georgescu-Mantoiu-JOT2001,Esteban-Loss-JMP2007,Voronov-Gitman-Tyutin-TMP2007,Arrizabalaga-JMP2011,Arrizabalaga-Duoandikoetxea-Vega_2012_JMP2013,Hogreve-2013_JPhysA,Esteban-Lewin-Sere-2017_DC-minmax-levels}.
A concise survey of this vast literature is discussed in \cite{Gallone-AQM2017}. Let us cast in Theorem \ref{thm:recap} here below the main relevant facts known today.
For the clarity of presentation, let us adopt natural units $c=\hbar=m=e=1$ henceforth, so as to get rid of mathematically inessential parameters, the coupling constant of relevance thus becoming $\nu\equiv - Z\alpha_\mathrm{f}$.

\begin{thm}[Self-adjoint extensions of the minimal Dirac-Coulomb]\label{thm:recap}
On the Hilbert space $\cH=L^2(\mathbb{R}^3,\mathbb{C}^4,\ud x)$ consider, for fixed $\nu\in\mathbb{R}$, the operator
\begin{equation}
 \begin{split}
  H\;&=\;H_0+\frac{\nu}{\,|x|\,}\mathbbm{1}\,,\qquad H_0\;=\;-\ii\,\bm{\alpha}\cdot \bm{\nabla}+\beta\,, \\
  \mathcal{D}(H)\;&=\;\mathcal{D}(H_0)\;=\;C^\infty_0(\mathbb{R}^3\!\setminus\!\{0\},\mathbb{C}^4)\,.
 \end{split}
\end{equation}
$H_0$ is essentially self-adjoint and the domain (of self-adjointness) of its operator closure $\overline{H_0}$ is $H^1(\mathbb{R}^3,\mathbb{C}^4)$. Moreover, the following holds.
\begin{itemize}
 \item[(i)] \emph{(Sub-critical regime.)} If $|\nu|\leqslant\frac{\sqrt{3}}{2}$, then $H$ is essentially self-adjoint and $\mathcal{D}(\overline{H})=H^1(\mathbb{R}^3,\mathbb{C}^4)$.
 \item[(ii)] \emph{(Critical regime.)} If $\frac{\sqrt{3}}{2}<|\nu|<1$, then $H$ admits an infinity of self-adjoint extensions, among which there is a `distinguished' one, $H_D$, uniquely characterised by the properties
 \begin{equation}
  \mathcal{D}(H_D)\,\subset\,\mathcal{D}(|H_0|^{1/2})\quad\mathrm{or}\quad  \mathcal{D}(H_D)\,\subset\,\mathcal{D}(|x|^{-1/2})\,,
 \end{equation}
 that is, the unique extension whose operator domain is both in the kinetic energy form domain $\mathcal{D}[H_0]=\mathcal{D}(|H_0|^{1/2})$ and in the potential energy form domain $\mathcal{D}[|x|^{-1}]=\mathcal{D}(|x|^{-1/2})$. Moreover, $0\notin\sigma(H_D)$.
 \item[(iii)] \emph{(Super-critical regime.)} If $|\nu|\geqslant 1$, then $H$ admits an infinity of self-adjoint extensions, without a distinguished one in the sense of the operator $H_D$ in the critical regime. In fact, when $|\nu|>1$ every self-adjoint extension of $H$ has infinitely many eigenfunctions not belonging to $\mathcal{D}(|x|^{-1/2})$.
\end{itemize}
In either regime, the spectrum of any self-adjoint extension $\widetilde{H}$ of $H$ is such that
\begin{equation}
 \begin{split}
  \sigma_\mathrm{ess}(\widetilde{H})\;&=\;\sigma(\overline{H_0})\;=\;(-\infty,-1]\cup[1,+\infty) \\
  \sigma_\mathrm{disc}(\widetilde{H})\;&\subset\;(-1,1)\,.
 \end{split}
\end{equation}
\end{thm}

It is worth remarking that for Coulomb-like matrix-valued interactions $V(x)$ that are \emph{not} of the form $\nu|x|^{-1}\mathbbm{1}$ but still satisfy  $|V(x)|\leqslant\;\nu|x|^{-1}$, the sub-critical regime described in Theorem \ref{thm:recap}(i) only ranges up to $|\nu|<\frac{1}{2}$, and counterexamples are well known of operators $H_0+V$ with $|V(x)|\leqslant\;(\frac{1}{2}+\varepsilon)|x|^{-1}$ for arbitrary $\varepsilon>0$ and failing to be essentially self-adjoint on $C^\infty_0(\mathbb{R}^3\!\setminus\!\{0\},\mathbb{C}^4)$ \cite{Arai-Yamada-RIMS-1983}.


In this work we are primarily focused on the critical regime, $|\nu|\in(\frac{\sqrt{3}}{2},1)$. This is a regime of ultra-heavy nuclei, in fact nuclei of elements that one expects to discover in the next future. It is the first regime where the Kato-Rellich-like perturbative arguments, applicable for small $\nu$'s, cease to work. It is also regarded as a physically meaningful regime, because as long as $|\nu|<1$ Sommerfeld's fine-structure formula still provides, formally, bound states for real energy levels, which only become complex when $|\nu|>1$, thus predicting an instability of the atom (the `$Z=137$ catastrophe').

In fact, the critical regime for the Dirac-Coulomb operator is already intensively studied, with a special focus on the `distinguished' self-adjoint extension $H_D$ \cite{Evans-1970,Weidmann-1971,Schmincke-1972-distinguished,Wust-1975,Nenciu-1976,Wust-1977,Klaus-Wust-1978,Landgren-Rejto-JMP1979,Arai-Yamada-RIMS-1983,Kato-1983,Xia-1999,Esteban-Loss-JMP2007,Voronov-Gitman-Tyutin-TMP2007,Arrizabalaga-JMP2011,Arrizabalaga-Duoandikoetxea-Vega_2012_JMP2013,Hogreve-2013_JPhysA,Esteban-Lewin-Sere-2017_DC-minmax-levels}. Conceptually, and qualitatively, this has very much in common with the analogous, scaling-critical problem of the self-adjoint realisation of the (formal) non-relativistic and pseudo-relativistic Schr\"{o}dinger operators
\[
 -\Delta+\frac{\nu}{\,|x|^2}\qquad\mathrm{or}\qquad \sqrt{1-\Delta}+\frac{\nu}{\,|x|\,}\qquad\mathrm{on}\;L^2(\mathbb{R}^3)
\]
when $\nu$ is out of the perturbative regime, 
an issue that is both standard textbook material \cite[Appendix to X.1]{rs2} and object of recent research \cite{LeYaouanc-Oliver-Raynal-JMP1997,B-Derezinski-G-AHP2011,Fall-Felli-JFA2014}.

Our perspective in the present work is that of a convenient \emph{classification} of all self-adjoint extensions of the minimally defined $H$, both in terms of explicit \emph{boundary conditions} for the functions in the domain of each extension, and in terms of an intrinsic, canonical structure of the domain of each extension. Moreover, unlike recent classifications \cite{Voronov-Gitman-Tyutin-TMP2007,Hogreve-2013_JPhysA} based on von Neumann's extension theory, we put the emphasis on the straightforward applicability of the so-called Kre{\u\i}n-Vi\v{s}ik-Birman extension theory \cite{GMO-KVB2017}, and in fact of its non-semi-bounded version, namely Grubb's universal classification theory \cite[Chapter 13]{Grubb-DistributionsAndOperators-2009}, which as a matter of fact turns out to be particularly versatile and informative in this context.

In order to give a first formulation of our main result, let us exploit, as customary, the canonical decomposition of $H$ into partial wave operators \cite[Section 4.6]{Thaller-Dirac-1992}, which is induced by its spherical symmetry. By expressing $x\equiv(x_1,x_2,x_3)\in\mathbb{R}^3$ in polar coordinates $x=(r,\Omega)\in\mathbb{R}^+\!\times\mathbb{S}^2$, $r:=|x|$, the map $\psi(x)\mapsto r\psi(x_1(r,\Omega),x_2(r,\Omega),x_3(r,\Omega))$ induces a unitary isomorphism
\[
 L^2(\mathbb{R}^3,\mathbb{C}^4,\ud x)\;\xrightarrow[]{\cong}\;L^2(\mathbb{R}^+,\ud r)\otimes L^2(\mathbb{S}^2,\mathbb{C}^4,\ud\Omega)\,.
\]
In terms of the observables
\[
\begin{split}
 \bm{L}&=\bm{x}\times(-\ii\bm{\nabla})\,,\qquad\qquad\quad\;\bm{S}=-\frac{1}{4}\bm{\alpha}\times\bm{\alpha}\,, \\
 \bm{J}&=\bm{L}+\bm{S}\equiv(J_1,J_2,J_3)\,,\quad K=\beta(2\bm{L}\cdot\bm{S}+\mathbbm{1})\,,
\end{split}
\]
one further decomposes
\begin{equation}
 L^2(\mathbb{S}^2,\mathbb{C}^4,\ud\Omega)\;\cong\;\bigoplus_{j\in\frac{1}{2}\mathbb{N}}\;\;\;\bigoplus_{m_j=-j}^j\;\bigoplus_{\kappa_j=\pm(j+\frac{1}{2})}\mathcal{K}_{m_j,\kappa_j}\,,
\end{equation}
where 
\begin{equation}
 \mathcal{K}_{m_j,\kappa_j}:=\;\mathrm{span}\{\Psi^+_{m_j,\kappa_j},\Psi^-_{m_j,\kappa_j}\}\;\cong\;\mathbb{C}^2
\end{equation}
and $\Psi^+_{m_j,\kappa_j}$ and $\Psi^-_{m_j,\kappa_j}$ are two orthonormal vectors in $\mathbb{C}^4$, and simultaneous eigenvectors of the observables $J^2\!\upharpoonright\! L^2(\mathbb{S}^2,\mathbb{C}^4,\ud\Omega)$, $J_3\!\upharpoonright\! L^2(\mathbb{S}^2,\mathbb{C}^4,\ud\Omega)$, and $K\!\upharpoonright\! L^2(\mathbb{S}^2,\mathbb{C}^4,\ud\Omega)$  with eigenvalue, respectively, $j(j+1)$, $m_j$, and $\kappa_j$. It then turns out that each subspace
\begin{equation}\label{eq:def_space_H_mj_kj}
 \cH_{m_j,\kappa_j}\;:=\;L^2(\mathbb{R}^+,\ud r)\otimes\mathcal{K}_{m_j,\kappa_j}\;\cong\;L^2(\mathbb{R}^+,\mathbb{C}^2,\ud r)
\end{equation}
is a reducing subspace for the Dirac-Coulomb Hamiltonian $H$, which, through the overall isomorphism
\begin{equation}
U\;:\;L^2(\mathbb{R}^3,\mathbb{C}^4,\ud x)\;\xrightarrow[]{\cong}\;\bigoplus_{j\in\frac{1}{2}\mathbb{N}}\;\;\;\bigoplus_{m_j=-j}^j\;\bigoplus_{\kappa_j=\pm(j+\frac{1}{2})}\cH_{m_j,\kappa_j}\,,
\end{equation}
is therefore unitarily equivalent to
\begin{equation}\label{eq:Dirac_operator_decomposition}
 UHU^*\;=\;\bigoplus_{j\in\frac{1}{2}\mathbb{N}}\;\;\;\bigoplus_{m_j=-j}^j\;\bigoplus_{\kappa_j=\pm(j+\frac{1}{2})}\;h_{m_j,\kappa_j}\,,
\end{equation}
where
\begin{equation}\label{eq:def_operator_h_mj_kj}
\begin{split}
  h_{m_j,\kappa_j}\;&:=\;\begin{pmatrix}
                   1+\frac{\nu}{r} & -\frac{\ud}{\ud r}+\frac{\kappa_j}{r} \\
                   \frac{\ud}{\ud r}+\frac{\kappa_j}{r} & -1+\frac{\nu}{r}
                  \end{pmatrix}, \\
 \qquad\mathcal{D}(h_{m_j,\kappa_j})\;&:=\;C^\infty_0(\mathbb{R}^+)\otimes \mathcal{K}_{m_j,\kappa_j}\;\cong\;C^\infty_0(\mathbb{R}^+,\mathbb{C}^2)\,.
\end{split}
\end{equation}
Thus, \eqref{eq:def_operator_h_mj_kj} defines a densely defined and symmetric operator on the Hilbert space \eqref{eq:def_space_H_mj_kj} and the overall problem of the self-adjoint realisation of $H$ is reduced to the same problem in each reducing subspace.

In particular, it is of physical relevance to consider each operator
\begin{equation}
 h_{m_j}\;:=\;h_{m_j,\kappa_j=j+\frac{1}{2}}\oplus h_{m_j,\kappa_j=-(j+\frac{1}{2})}
\end{equation}
acting block-diagonal-wise, with the two different spin-orbit components, on the Hilbert eigenspace $L^2(\mathbb{R}^+,\mathbb{C}^4,\ud r)$ of $(j,m_j)$-eigenvalue for $J^2$ and $J_3$.

Now, the following property is well known, as one can see by means of standard limit-point limit-circle arguments \cite[Chapter 6.B]{Weidmann-book1987}. Its proof is discussed, e.g., in \cite[Section 2]{Gallone-AQM2017}.

\begin{prop}\label{prop:deficiency_indices}
 The operator $h_{m_j,\kappa_j}$ is essentially self-adjoint on its domain with respect to the Hilbert space $\cH_{m_j,\kappa_j}$ if and only if
 \begin{equation}
  \nu^2\;\leqslant\;\kappa_j^2-\textstyle{\frac{1}{4}}\,,
 \end{equation}
 and it has deficiency indices $(1,1)$ otherwise. In particular, in the regime $|\nu|\in(\frac{\sqrt{3}}{2},1)$ only the operators of the decomposition \eqref{eq:Dirac_operator_decomposition} with $\kappa_j^2=1$, thus
 \begin{equation}\label{eq:4op}
  h_{\frac{1}{2},1}\,,\quad h_{-\frac{1}{2},1}\,,\quad h_{\frac{1}{2},-1}\,,\quad h_{-\frac{1}{2},-1}\,,
 \end{equation}
 have deficiency indices $(1,1)$, all others being essentially self-adjoint.
\end{prop}

Therefore the operator $h_{\frac{1}{2},1}\oplus h_{\frac{1}{2},-1}\oplus h_{-\frac{1}{2},1}\oplus h_{-\frac{1}{2},-1}$, and hence $H$ itself, has deficiency indices $(4,4)$. 
This means that there is a 16-real-parameter family of self-adjoint extensions of $H$, hence of physically inequivalent realisations of the Dirac-Coulomb Hamiltonian. From the operator-theoretic point of view, the analysis of the self-adjoint extensions of $h_{\frac{1}{2},1}$ is the very same as for the other three operators (and in fact $h_{\frac{1}{2},1}$ and $h_{-\frac{1}{2},1}$ have the same formal action on $L^2(\mathbb{R}^+,\mathbb{C}^2)$, and so have $h_{\frac{1}{2},-1}$ and $h_{-\frac{1}{2},-1}$), and hence we will discuss only the first case.

There is room for extensions only on the sector $j=\frac{1}{2}$ of lowest total angular momentum $J^2$ and, as we shall discuss in Section \ref{sec:mainresults}, each extension corresponds to a particular prescription on the wave functions of the domain in the vicinity of the centre $x=0$ of the Coulomb interaction. For higher $j$'s the large angular momentum makes the Coulomb singularity lesser and lesser relevant, and on such sectors $H$ is already essentially self-adjoint.

Physically, the relevant class of extensions is rather the \emph{one}-parameter sub-family consisting of the same extension for each elementary operators \eqref{eq:4op}, in a sense that will be evident in the next Section, that is, extensions where the same boundary conditions of self-adjointness occurs on each block of $H$ -- it would be non-physical to have a different behaviour of the physical Hamiltonian on different sectors $\mathcal{H}_{m_j,\kappa_j}$ of its symmetry.

We can now conclude this Introduction by anticipating an informal version of the main results that we will present rigorously in Section \ref{sec:mainresults}. First and foremost, we do \emph{not} apply the self-adjoint extension theory of von Neumann, unlike what is done ubiquitously in the previous literature, and we exploit instead (and to our knowledge for the first time) the Kre{\u\i}n-Vi\v{s}ik-Birman / Grubb extension scheme. To our taste such a scheme produces in this context the most informative version of the classification of the self-adjoint extensions of $H$ in a considerably less laborious way.

Informally, our results can be summarised as follows.

\begin{thm}[Classification of Dirac-Coulomb extensions -- informal version]~ \newline

\vspace{-0.7cm}
\noindent Let $|\nu|\in(\frac{\sqrt{3}}{2},1)$.
\begin{itemize}
 \item[(i)] On each of the four sectors $(j,m_j,\kappa_j)=(\frac{1}{2},\pm\frac{1}{2},\pm 1)$ of non-self-adjoint\-ness, the operator $H$ admits a one-parameter family $(S_\beta)_{\beta\in\mathbb{R}\cup\{\infty\}}$ of self-adjoint extensions -- which are then restrictions of $H^*$.
 \item[(ii)] Whereas the domain of $H^*$ in each sector consists of spinors $g$ with $H^1$-regularity on $[\varepsilon,+\infty)$ for all $\varepsilon>0$ and the short-distance asymptotics
\begin{equation}\label{eq:g_asympt}
\begin{split}
 g(r)\;&=\; g_0 r^{-\sqrt{1-\nu^2}} + g_1 r^{\sqrt{1-\nu^2}} + o(r^{1/2})\quad\textrm{as }\;r\downarrow 0 
\end{split}
\end{equation}
for some $g_0,g_1\in\mathbb{C}^2$ dependent on $\nu$ only, 
the domain $\mathcal{D}(S_\beta)$ of the extension $S_\beta$ consists of those such spinors for which a prescribed ratio holds between the corresponding components of $g_1$ and $g_0$, for concreteness 
\begin{equation}
\frac{g_1^+}{g_0^+}\;=\;c_\nu\,\beta+d_\nu
\end{equation}
for some explicitly known constants $c_\nu,d_\nu\in\mathbb{C}$.
 \item[(iii)] The extension $\beta=\infty$ is the restriction $S_D$, on the considered sector $(j,m_j,\kappa_j)$, of the distinguished extension $H_D$ of $H$ discussed in Theorem \ref{thm:recap}(ii): the functions in its domain have the asymptotics \eqref{eq:g_asympt} with $g_0\equiv 0$, i.e., without singular term.
 \item[(iv)] All those extensions $S_\beta$ with $\beta\neq 0$ are invertible with everywhere defined and bounded inverse, in which case the inverse $S_\beta^{-1}$ is an explicit rank-one perturbation of $S_D^{-1}$.
 \item[(v)] The gap in the spectrum of $S_\beta$ around $\lambda=0$ has a direct estimate in terms of $\beta$ and $\|S_D^{-1}\|$ and must be at least the interval
\begin{equation}
 (-E_0(\beta),E_0(\beta))\,,\qquad E_0(\beta)\,:=\,\frac{|\beta|}{|\beta|\|S_D^{-1}\|+1}\,.
\end{equation}
\end{itemize}
\end{thm}

Last, here is how the material is organised. As mentioned already, in Section \ref{sec:mainresults} we state rigorously our main results, whose proof, outlined in Section \ref{sec:mainresults} itself, is based on intermediate results that we prove in Sections \ref{sec:kernel}, \ref{sec:distinguished}, and \ref{sec:closure}. In Section \ref{sec:resolvent} we discuss further properties of the Dirac-Coulomb extensions involving the resolvent and the spectral gap at zero.

\textbf{Notation.} Essentially all the notation adopted here is standard, let us only emphasize the following. Concerning the various sums of spaces that will occur, we denote by $\dotplus$ the direct sum of vector spaces, by $\oplus$ the direct orthogonal sum of \emph{closed} Hilbert subspaces of the same underlying Hilbert space,
and by $\boxplus$ the direct sum of subspaces of $\cH$ that are orthogonal to each other but are not a priori all closed. 
\emph{Operator} domain and \emph{form} domain of any given densely defined and symmetric operator $S$ are denoted, respectively, by $\mathcal{D}(S)$ and $\mathcal{D}[S]$. 
As customary, $\mathbb{R}^+=(0,+\infty)$, and $\sigma(T)$ and $\rho(T)$ denote, respectively, the spectrum and the resolvent set of an operator $T$ on Hilbert space. The notation $A\lesssim B$ stands for the inequality $A\leqslant cB$ for some constant that is universal or is clear from the context not to depend on the other variables of the inequality itself. We refer to  elements $g\in\mathbb{C}^2$ or in $L^2(\mathbb{R}^+,\mathbb{C}^2)$ as spinors $g=\begin{pmatrix} g^+\! \\ g^-\! \end{pmatrix}$. $G^T$ denotes the transpose of a matrix $G$.

\section{Classification scheme and main results}\label{sec:mainresults}

The original Kre{\u\i}n-Vi\v{s}ik-Birman scheme \cite{Alonso-Simon-1980,GMO-KVB2017} for the determination and classification of the self-adjoint extensions of a given densely defined and symmetric operator on Hilbert space was developed for \emph{semi-bounded} operators: for this case one can non-restrictively assume that the bottom of the operator $S$ to extend is strictly positive and hence a canonical extension exists, the Friedrichs extension $S_F$, with the same bottom and hence with everywhere defined bounded inverse $S_F^{-1}$.

In fact, to a large extent, the role of $S_F$ in the theory can be played as well by any other `distinguished' self-adjoint extension $S_D$ of $S$ which is itself invertible with everywhere defined and bounded inverse $S_D^{-1}$, and this makes many results of the theory applicable also to a (densely defined and symmetric) non-semi-bounded $S$, provided that $S$ admits such an extension $S_D$. In this spirit, Grubb's `universal classification' scheme was later developed \cite{Grubb-1968} (a modern survey of which may be found in \cite[Chapter 13]{Grubb-DistributionsAndOperators-2009}), which only makes reference to the existence of an invertible extension and, in the case of symmetric operators, it reproduces many features of the Kre{\u\i}n-Vi\v{s}ik-Birman scheme.

For our next purposes, let us single out from such extension theories the following result, for a discussion of which we refer to \cite[Theorem 3.4]{GMO-KVB2017}.

\begin{thm}[Classification of self-adjoint extensions -- operator version]\label{thm:VB-representaton-theorem_Tversion} 
Let $S$ be a densely defined symmetric operator on a Hilbert space $\cH$ and assume that $S$ admits a self-adjoint extension $S_D$ which is invertible with everywhere defined inverse.
Then there is a one-to-one correspondence between the  family of all self-adjoint extensions of  $S$ on $\cH$ and the family of the self-adjoint operators on Hilbert subspaces of $\ker S^*$.
If $T$ is any such operator, in the correspondence $T\leftrightarrow S_T$ each self-adjoint extension $S_T$ of $S$ is given by
\begin{equation}\label{eq:ST}
\begin{split}
S_T\;&=\;S^*\upharpoonright\mathcal{D}(S_T) \\
\mathcal{D}(S_T)\;&=\;\left\{f+S_D^{-1}(Tv+w)+v\left|\!\!
\begin{array}{c}
f\in\mathcal{D}(\overline{S})\,,\;v\in\mathcal{D}(T) \\
w\in\ker S^*\cap\mathcal{D}(T)^\perp
\end{array}\!\!
\right.\right\}.
\end{split}
\end{equation}
\end{thm}

When applying Theorem \ref{thm:VB-representaton-theorem_Tversion}  to the extension problem for the operator $h_{m_j,\kappa_j}$ defined in  \eqref{eq:def_operator_h_mj_kj} on the Hilbert space $\cH_{m_j,\kappa_j}$ defined in \eqref{eq:def_space_H_mj_kj}, it is natural that the reference extension is the distinguished extension of Theorem \ref{thm:recap}(ii).

%
%

Acting on the Hilbert space $L^2(\mathbb{R}^+,\mathbb{C}^2,\ud r)$ with scalar product
\begin{equation}\label{eq:scalar_product}
 \begin{split}
  \langle \psi,\phi\rangle_{L^2(\mathbb{R}^+,\mathbb{C}^2)}\,&=\,\int_0^{+\infty}\langle\psi(r),\phi(r)\rangle_{\mathbb{C}^2}\,\ud r \,=\,\sum_{\alpha=\pm}\int_0^{+\infty}\overline{\psi^\alpha(r)}\,\phi^\alpha(r)\,\ud r \\
  &\quad\psi\equiv\begin{pmatrix}
                   \psi^+\! \\ \psi^-\!
                  \end{pmatrix},\,
        \phi\equiv\begin{pmatrix}
                   \phi^+\! \\ \phi^-\!
                  \end{pmatrix}\,\in\,L^2(\mathbb{R}^+,\mathbb{C}^2)\,,
 \end{split} 
\end{equation}
we consider the operator
\begin{equation}\label{eq:def_operator_S}
\begin{split}
  S\;&:=\;\begin{pmatrix}
                   1+\frac{\nu}{r} & -\frac{\ud}{\ud r}+\frac{1}{r} \\
                   \frac{\ud}{\ud r}+\frac{1}{r} & -1+\frac{\nu}{r}
                  \end{pmatrix}, \\
 \qquad\mathcal{D}(S)\;&:=\;C^\infty_0(\mathbb{R}^+,\mathbb{C}^2)\,.
\end{split}
\end{equation}
$S$ is non-semi-bounded, densely defined, and symmetric, and following from Proposition \ref{prop:deficiency_indices} it has deficiency indices $(1,1)$. For convenience, let us also denote by $\widetilde{S}$ the differential operator defined by
\begin{equation}\label{eq:tildeS}
 \widetilde{S}
 \begin{pmatrix}
  f^+\\ f^-
 \end{pmatrix}\;:=\;\begin{pmatrix}
                   1+\frac{\nu}{r} & -\frac{\ud}{\ud r}+\frac{1}{r} \\
                   \frac{\ud}{\ud r}+\frac{1}{r} & -1+\frac{\nu}{r}
                  \end{pmatrix}\begin{pmatrix}
  f^+ \\ f^-
 \end{pmatrix}.
\end{equation}
Since $\widetilde{S}$ has real smooth coefficients, and is formally self-adjoint, it is a standard fact \cite[Section 4.1]{Grubb-DistributionsAndOperators-2009} that the operator closure $\overline{S}$ and the adjoint $S^*$ of $S$ are nothing but, respectively, the \emph{minimal} and the \emph{maximal realisation} of $\widetilde{S}$, that is, they both act as $\widetilde{S}$ respectively on
\begin{equation}\label{DSclosureDS*}
 \begin{split}
  \mathcal{D}(\widetilde{S})\;&=\;\overline{C^\infty_0(\mathbb{R}^+,\mathbb{C}^2)}^{\|\cdot\|_S} \\
  \mathcal{D}(S^*)\;&=\;\{\psi\in L^2(\mathbb{R}^+,\mathbb{C}^2)\,|\,\widetilde{S}\psi \in L^2(\mathbb{R}^+,\mathbb{C}^2)\}\,,
 \end{split}
\end{equation}
where $\|\cdot\|_S$ is the graph norm associated with $S$. One has $\overline{S}\subset S^*$, and by \emph{self-adjoint realisation} of $S$ we shall mean any operator $R=R^*$ on $L^2(\mathbb{R}^+,\mathbb{C}^2)$ such that $\overline{S}\subset R\subset S^*$.

In order to identify the self-adjoint realisations of $S$ using the Kre{\u\i}n-Vi\v{s}ik-Birman scheme of Theorem \ref{thm:VB-representaton-theorem_Tversion}, we shall collect the intermediate results of Propositions \ref{prop:kerS*}, \ref{prop:SD}, and \ref{prop:Sclosure} below, whose proof is deferred to the following Sections.

For convenience, let us introduce the parameter
\begin{equation}
 B\;:=\;\sqrt{1-\nu^2}\,.
\end{equation}
It will be important to remember throughout our analysis that $B\in(0,\textstyle{\frac{1}{2}})$.

First one needs a characterisation of $\ker S^*$.

\begin{prop}\label{prop:kerS*}
 For every $|\nu|\in(\frac{\sqrt{3}}{2},1)$ the operator $S^*$ has a one dimensional kernel, spanned by the function
 \begin{equation}\label{eq:vector_Phi}
  \Phi\:=\:\begin{pmatrix}
            \Phi^+\! \\ \Phi^-
           \!\end{pmatrix}
 \end{equation}
 with
 \begin{equation}\label{eq:vector_Phi_components}
  \Phi^{\pm}(r)\;:=\;e^{-r}r^{-B}\big( \textstyle{\frac{\pm(1+\nu)+B}{1+\nu}}\,U_{-B,1-2B}(2r)-\textstyle{\frac{2rB}{1+\nu}}\,U_{1-B,2-2B}(2r)\big)\,,
 \end{equation}
where
$U_{a,b}(r)$ is the Tricomi function \cite[Sec.~13.1.3]{Abramowitz-Stegun-1964}. $\Phi$ is analytic on $(0,+\infty)$ with asymptotics
\begin{equation}\label{eq:Phi_asymptotics}
 \begin{split}
  \Phi(r)\;&=\;r^{-B}\,\textstyle{\frac{\Gamma(2B)}{\Gamma(B)}}
  \begin{pmatrix}
   \;\frac{1+\nu+B}{1+\nu} \\
   -\frac{1+\nu-B}{1+\nu}
  \end{pmatrix}+
  \begin{pmatrix}
   q^+\! \\ q^-\!
  \end{pmatrix}r^B+O(r^{1-B})\quad\textrm{as }\;r\downarrow 0 \\
   \Phi(r)\;&=\;2^B\begin{pmatrix}
               1 \\ -1
              \end{pmatrix}r^{-B}e^{-r}
              (1+O(r^{-1}))\quad\,\textrm{as }\;r\to +\infty\,,
 \end{split}
\end{equation}
where $q^\pm$ are both non-zero and explicitly given by \eqref{eq:def_qpm} below.
\end{prop}

Next, one needs to identify a reference extension $S_D$ of $S$ which be self-adjoint and with everywhere defined inverse, and to characterise the action of $S_D$ on $\ker S^*$.

\begin{prop}\label{prop:SD}~
\begin{itemize}
 \item[(i)]  There exists a self-adjoint realisation $S_D$ of $S$ with the property that
 \begin{equation}\label{eq:SD_uniqueness_properties}
  \mathcal{D}(S_D)\subset H^{1/2}(\mathbb{R}^+,\mathbb{C}^2)\quad\textrm{or}\quad \mathcal{D}(S_D)\subset\mathcal{D}[r^{-1}]\,,
 \end{equation}
 where the latter is the form domain of the multiplication operator by $r^{-1}$ on each component of $L^2(\mathbb{R}^+,\mathbb{C}^2)$ (the space of `finite potential energy'). $S_D$ is the only self-adjoint realisation of $S$ satisfying \eqref{eq:SD_uniqueness_properties}.
 \item[(ii)] $S_D$ is invertible on $L^2(\mathbb{R}^+,\mathbb{C}^2)$ with everywhere defined and bounded inverse. The \emph{explicit} integral kernel of $S_D$ is given by \eqref{eq:Green}.
 \item[(iii)] In terms of the spaces $\mathcal{D}(\overline{S})$ and $\ker S^*$ one has
 \begin{equation}\label{eq:decomp_DSD}
  \mathcal{D}(S_D)\;=\;\mathcal{D}(\overline{S})\dotplus S_D^{-1}\ker S^*\,.
 \end{equation}
 Moreover,
 \begin{equation}\label{eq:Krein_decomp_formula}
  \begin{split}
  \mathcal{D}(S^*)\;&=\; \mathcal{D}(S_D)\dotplus\ker S^*\,,\\
  &=\;\mathcal{D}(\overline{S})\dotplus S_D^{-1}\ker S^*\dotplus\ker S^*\,.
  \end{split}
 \end{equation}
 \item[(iv)] For the vector $S_D^{-1}\Phi$, where $\Phi\in\ker S^*$ is given by \eqref{eq:vector_Phi}-\eqref{eq:vector_Phi_components}, one has the following point-wise asymptotics:
 \begin{equation}\label{SDPhi_vanishing}
  S_D^{-1}\Phi(r)\;\sim\;\begin{pmatrix} p^+\! \\ p^-\!\end{pmatrix} r^B+o(r^{1/2})\qquad \textrm{as }\;r\downarrow 0\,,
 \end{equation}
\end{itemize}
where  $p^{\pm}$ are both non-zero and explicitly given in \eqref{ed:defppm} below.\footnote{In fact, with a slightly more elaborate argument we can better estimate the reminder in \eqref{SDPhi_vanishing} as a $O(r^{1-B})$ term; however, this is not needed in the analysis that follows.}
\end{prop}

Last, an amount of information is needed on the domain of the operator closure $\overline{S}$ of $S$. Although $\mathcal{D}(\overline{S})$ is canonically constructed as the closure of $\mathcal{D}(S)$ in the operator norm, it does not correspond to a standard functional space. In fact in Section \ref{sec:closure} we will present a complete characterisation of $\mathcal{D}(\overline{S})$, from which we will be able to deduce the following properties, relevant for our main results.

\begin{prop}\label{prop:Sclosure}~
Let $f\in\mathcal{D}(\overline{S})$. Then $f\in H^1_{\mathrm{loc}}(\mathbb{R}^+)$ and 
\begin{equation}\label{eq:f_vanishing}
 f(r)\;=\;o(r^{1/2})\qquad \textrm{as } \quad r\downarrow 0\,.
\end{equation}
\end{prop}

With Propositions \ref{prop:kerS*}, \ref{prop:SD}, and \ref{prop:Sclosure} at hand, we can now formulate a general classification as follows.

\medskip

\begin{thm}[Classification of the self-adjoint realisations for the Dirac-Coulomb Hamiltonian -- structural version]\label{thm:classification_structure}~

\noindent The self-adjoint extensions of the operator $S$ on $L^2(\mathbb{R}^+,\mathbb{C}^2,\ud r)$ defined in \eqref{eq:def_operator_S} constitute a one-parameter family $(S_\beta)_{\beta\in\mathbb{R}\cup{\{\infty\}}}$ of restrictions of the adjoint operator $S^*$ determined in \eqref{DSclosureDS*}, each of which is given by
\begin{equation}\label{eq:Sbeta}
 \begin{split}
  S_\beta\;&:=\;S^*\upharpoonright\mathcal{D}(S_\beta) \\
  \mathcal{D}(S_\beta)\;&:=\;\left\{g=f+c(\beta S_D^{-1}\Phi+\Phi)\left|\!
  \begin{array}{c}
   f\in\mathcal{D}(\overline{S}) \\
   c\in\mathbb{C}
  \end{array}\right.\!\!\right\}.
 \end{split}
\end{equation}
Here $S_D$ is the distinguished self-adjoint extension of $S$ identified in Proposition \ref{prop:SD} and $\Phi$ is the spanning element of $\ker S^*$ identified in Proposition \ref{prop:kerS*}. In this parametrisation the distinguished extension $S_D$ corresponds to $\beta=\infty$. For each $g\in\mathcal{D}(S_\beta)$ the function $f\in\mathcal{D}(\overline{S})$ and the constant $c\in\mathbb{C}$ are uniquely determined. 
\end{thm}

Mirror to the parametrisation formula \eqref{eq:Sbeta}, we can re-express the above result in terms of boundary conditions at the centre of the Coulomb singularity.

\begin{thm}[Classification of the self-adjoint realisations for the Dirac-Coulomb Hamiltonian -- boundary condition version]\label{thm:classification_bc}~

\begin{itemize}
 \item[(i)] Any function $g=\begin{pmatrix} g^+ \\ g^-\end{pmatrix}\in\mathcal{D}(S^*)$ satisfies the short-distance asymptotics
 \begin{equation}\label{eq:coeff_a_b}
 \begin{split}
  \lim_{r\downarrow 0} \,r^B g(r)\;&=\;g_0 \\
  \lim_{r\downarrow 0} \,r^{-B}(g(r)-g_0r^{-B})\;&=\;g_1
 \end{split}
 \end{equation}
 for some $g_0,g_1\in\mathbb{C}^2$.
 In particular,
 \begin{equation}\label{eq:coeff_a_b_BIS}
  g(r)\;=\;g_0\, r^{-B}+g_1r^B+o(r^{1/2})\qquad\textrm{as }\;r\downarrow 0\,.
 \end{equation}
 \item[(ii)] The self-adjoint extensions of the operator $S$ on $L^2(\mathbb{R}^+,\mathbb{C}^2)$ defined in \eqref{eq:def_operator_S} constitute a one-parameter family $(S_{\beta})_{\beta\in\mathbb{R}\cup{\{\infty\}}}$ of restrictions of the adjoint operator $S^*$, each of which is given by
\begin{equation}\label{eq:Sbeta_bc}
 \begin{split}
  S_{\beta}\;&:=\;S^*\upharpoonright\mathcal{D}(S_{\beta}) \\
  \mathcal{D}(S_{\beta})\;&:=\;\Big\{g\in\mathcal{D}(S^*)\,\Big|\,\frac{g_1^+}{g_0^+}=c_\nu \beta+d_\nu\Big\}\,,
 \end{split}
\end{equation}
where
\begin{equation}\label{eq:defcd}
 \begin{split}
  c_\nu\;&=\;p^+{\textstyle\Big(\frac{\Gamma(2B)}{\Gamma(B)}\,\frac{1+\nu+B}{1+\nu}\Big)^{\!-1}} \\
  d_\nu\;&=\;q^+{\textstyle\Big(\frac{\Gamma(2B)}{\Gamma(B)}\,\frac{1+\nu+B}{1+\nu}\Big)^{\!-1}},
 \end{split}
\end{equation}
and $p^+$ and $q^+$ are given, respectively, by \eqref{ed:defppm} and \eqref{eq:def_qpm}.
This is precisely the same parametrisation of the extension as in Theorem \ref{thm:classification_structure}.

\end{itemize}
\end{thm}

The proofs of Theorems \ref{thm:classification_structure} and \ref{thm:classification_bc} are an application of the general Kre{\u\i}n-Vi\v{s}ik-Birman Theorem \ref{thm:VB-representaton-theorem_Tversion}, through the intermediate results of Propositions  \ref{prop:kerS*}, \ref{prop:SD}, and \ref{prop:Sclosure}, as we shall show in a moment. Owing to further corollaries of Theorem \ref{thm:VB-representaton-theorem_Tversion}, which we work out in detail in Section \ref{sec:resolvent} (Theorem \ref{eq:thm_invertibility_and_resolvent} therein), we can supplement the above extension picture with an additional amount of information concerning the invertibility, the resolvent, and the spectral gap of each realisation $S_\beta$. This too is an example of relevant and non-trivial features of the self-adjoint extensions that can be established in a relatively cheap and elementary manner, unlike the counterpart way based on von Neumann's extension theory.

\begin{thm}[Invertibility, resolvent, and estimate on the spectral gap]\label{thm:DC-invertibility-resolvent-gap}~

\noindent The elements of the family $(S_\beta)_{\beta\in\mathbb{R}\cup{\{\infty\}}}$ of the self-adjoint extensions of the operator $S$ on $L^2(\mathbb{R}^+,\mathbb{C}^2,\ud r)$ defined in \eqref{eq:def_operator_S}, labelled according to the parametrisation of Theorem \ref{thm:classification_structure}, have the following properties.
\begin{itemize}
 \item[(i)] $S_\beta$ is invertible on the whole $L^2(\mathbb{R}^+,\mathbb{C}^2)$ if and only if $\beta\neq 0$.
 \item[(ii)] For each invertible extension $S_\beta$,
 \begin{equation}\label{eq:Sbeta-1}
  S_\beta^{-1}\;=\;S_D^{-1}+\frac{1}{\,\beta\|\Phi\|^{2}}\:|\Phi\rangle\langle\Phi|\,.
 \end{equation}
 \item[(iii)] For each invertible extension $S_\beta$,
 \begin{equation}\label{eq:sigmaess}
  \sigma_{\mathrm{ess}}(S_\beta)\;=\;\sigma_{\mathrm{ess}}(S_D)\;=\;(-\infty,-1]\cup[1,+\infty)\,,
 \end{equation}
 and the gap in the spectrum $\sigma(S_\beta)$ around $E=0$ is at least the interval $(-E(\beta),E(\beta))$, where
 \begin{equation}
  E(\beta)\;:=\;\frac{|\beta|}{\,|\beta| \|S_D^{-1}\|+1\,}\,.
 \end{equation}
\end{itemize}
\end{thm}

We conclude this Section with the proof of Theorems \ref{thm:classification_structure} and \ref{thm:classification_bc}, and we defer the proof of the technical intermediate results and of Theorem \ref{thm:DC-invertibility-resolvent-gap} to the following Sections.

\begin{proof}[Proof of Theorem \ref{thm:classification_structure}]
One extension is surely the distinguished extension $S_D$, with domain $\mathcal{D}(S_D)=\mathcal{D}(\overline{S})\dotplus S_D^{-1}\ker S^*$ (Proposition \ref{prop:SD}(iii)), which is of the form \eqref{eq:Sbeta} for $\beta=\infty$: indeed, with respect to the general formula \eqref{eq:ST}, this is the extension that corresponds to an operator $T$ defined on $\{0\}\subset\ker S^*$.
Since $\dim\ker S^*=1$ (Proposition \ref{prop:kerS*}), for all other extensions of $S$ the parametrising operator $T$, in the sense of the general formula \eqref{eq:ST}, must be self-adjoint on the whole one-dimensional $\mathrm{span}\{\Phi\}$, and therefore is the multiplication operator by a scalar $\beta$. Then \eqref{eq:ST} takes the form \eqref{eq:Sbeta}. The uniqueness of the decomposition of $g\in\mathcal{D}(S_\beta)$ into $g\in\mathcal{D}(S_\beta)$ is a direct consequence of the direct sum decomposition \eqref{eq:Krein_decomp_formula} of Proposition \ref{prop:SD}(iii).
%
\end{proof}

\medskip

\begin{proof}[Proof of Theorem \ref{thm:classification_bc}]~

(i) It was determined in Propositions \ref{prop:kerS*}, \ref{prop:SD}, and \ref{prop:Sclosure} that a generic $g\in\mathcal{D}(S^*)$ decomposes with respect to \eqref{eq:Krein_decomp_formula} as
\begin{equation}\label{eq:decomp_g_DS*}
\begin{pmatrix} g^+ \\ g^-\end{pmatrix}\;=\;\begin{pmatrix} f^+ \\ f^-\end{pmatrix}+a\, S_D^{-1}\begin{pmatrix} \Phi^+ \\ \Phi^-\end{pmatrix}+\frac{b}{\gamma}\begin{pmatrix} \Phi^+ \\ \Phi^-\end{pmatrix}\quad \gamma\,:=\,\textstyle{\frac{\Gamma(2B)}{\Gamma(B)}\,\frac{1+\nu+B}{1+\nu}}
\end{equation}
for some $a,b\in\mathbb{C}$, and moreover, as $r\downarrow 0$,
\[
 \begin{split}
  f(r)\;&=\;o(r^{1/2})\,, \\
  (S_D^{-1}\Phi)(r)\;&=\; \begin{pmatrix} p^+\! \\ p^-\!\end{pmatrix} r^B+o(r^{1/2})\,, \\
  r^B\Phi(r)\;&=\;\begin{pmatrix} 1 \\ -\frac{1+\nu-B}{1+\nu+B}\end{pmatrix}
  \gamma+\begin{pmatrix}
    q^+\! \\q^-\!
   \end{pmatrix} r^{2B}+o(r^{1/2+B})
 \end{split}
\]
(see, respectively, \eqref{eq:f_vanishing}, \eqref{SDPhi_vanishing}, and \eqref{eq:Phi_asymptotics} above). Therefore, the limit in the first component yields $r^Bg^+\!(r)\xrightarrow[]{r\downarrow 0}b$, and also
\[
 \begin{split}
  r^{-B}&(g^+\!(r)-br^{-B})\;=\;r^{-B}\big(f^+\!(r)+a(S_D^{-1}\Phi)^+\!(r)+b\gamma^{-1}\Phi^+\!(r)-br^{-B}\big) \\
  &=\;a\,p^+ +b\,q^+\gamma^{-1} + o(r^{1/2-B})\,,
 \end{split}
\]
that is, $r^{-B}(g^+\!(r)-br^{-B})\xrightarrow[]{r\downarrow 0}a\,p^+ +b\,q^+\gamma^{-1}$. Thus, \eqref{eq:coeff_a_b} and \eqref{eq:coeff_a_b_BIS} follow by setting
\[\tag{*}
 g_0^+\;:=\;b\,,\qquad g_1^+\;:=\;a p^+ + b q^+\gamma^{-1}
\]
%
(an analogous argument holds for the lower components).

(ii) Necessary and sufficient condition for $g\in\mathcal{D}(S^*)$ to belong to the domain $\mathcal{D}(S_\beta)$ of the extension $S_\beta$ determined by \eqref{eq:Sbeta} of Theorem \ref{thm:classification_structure} is that in the decomposition \eqref{eq:decomp_g_DS*} above the coefficients $a$ and $b$ satisfy $a=\beta b\gamma^{-1}$. Owing to (*) and \eqref{eq:defcd}, the latter condition reads $g_1^+/g_0^+=c_\nu \beta+d_\nu$.
\end{proof}

Last, it is worth highlighting a couple of important remarks.

\begin{rem}
The proof of Theorem \ref{thm:classification_bc} shows that the decomposition of $g\in\mathcal{D}(S_\beta)$ determined by \eqref{eq:Sbeta}, and hence $c$ and $f$, are explicitly given by
 \begin{equation}\label{eq:comp_of_g}
 \begin{split}
  c\;&=\;(\textstyle{\frac{\Gamma(B)}{\Gamma(2B)}\,\frac{1+\nu}{1+\nu+B}})\cdot\displaystyle{\lim_{r\downarrow 0}} \;r^B g^+\!(r) \\
  f\;&=\;g-c(\beta S_D^{-1}\Phi+\Phi)\,.
 \end{split}
 \end{equation}
 Indeed, in the notation of \eqref{eq:decomp_g_DS*} therein, $b=\gamma c$. In fact, the same argument shows that the first equation in \eqref{eq:comp_of_g} determines the component $c\Phi\in\ker S^*$ of a generic $g\in\mathcal{D}(S^*)$, and hence defines the (non-orthogonal) projection $\mathcal{D}(S^*)\to\ker S^*$, $g\mapsto c\Phi$ induced by the decomposition formula \eqref{eq:Krein_decomp_formula}.
When $\beta\neq 0$, one has equivalently 
\begin{equation}\label{eq:comp_of_g_alt}
 c\;=\;\beta^{-1}\!\!\int_0^{+\infty}\!\langle\Phi(r),(\widetilde{S}g)(r)\rangle_{\mathbb{C}^2}\,\ud r \,.
\end{equation}
Indeed $\widetilde{S}g=S_\beta g=S^*g=\overline{S}f+c\beta\Phi$ and $\ran\overline{S}\perp\ker S^*$, whence it follows that $\langle\Phi,\widetilde{S}g\rangle_{L^2(\mathbb{R}^3,\mathbb{C}^2)}=c\beta$. 
\end{rem}

\begin{rem}
 As typical when the operator which one studies the self-adjoint extensions of is a differential operator, one interprets \eqref{eq:Krein_decomp_formula} as the canonical decomposition of an element $g\in\mathcal{D}(S^*)$ into a `regular' and a `singular' part
 \begin{equation}
  \begin{split}
   g_{\mathrm{reg}}\;&:=\;f+a\, S_D^{-1}\Phi\in\mathcal{D}(S_D) \\
   g_{\mathrm{sing}}\;&:=\;\frac{b}{\gamma}\,\Phi\in\ker S^*\,,
  \end{split}
 \end{equation}
 where $a,b\in\mathbb{C}$ and $f\in\mathcal{D}(\overline{S})$ are determined by $g$ and $\gamma=\frac{\Gamma(2B)}{\Gamma(B)}\,\frac{1+\nu-B}{1+\nu}$. Indeed $\mathcal{D}(S_D)$ has a higher regularity then $\ker S^*$: functions in the former space vanish at zero, as follows from \eqref{SDPhi_vanishing}-\eqref{eq:f_vanishing}, whereas $\Phi$ diverges at zero, as seen in \eqref{eq:Phi_asymptotics}. In this language, $r^{-B}g_{\mathrm{reg}}^+(r)\xrightarrow[]{r\downarrow 0}a p^+$ and $r^{B}g_{\mathrm{sing}}^+(r)\xrightarrow[]{r\downarrow 0}b$, and the self-adjointness condition \eqref{eq:Sbeta_bc} that selects, among the elements in $\mathcal{D}(S^*)$, only those in $\mathcal{D}(S_{\beta})$ reads
 \begin{equation}\label{eq:bc-reg-sing}
  \Big(\frac{\gamma}{p^+}\lim_{r\downarrow 0}\;r^{-B}g_{\mathrm{reg}}^+(r)\Big) \;=\; (c_\nu \beta+d_\nu)\,\Big(\lim_{r\downarrow 0}\;r^{B}g_{\mathrm{sing}}^+(r)\Big)\,,
 \end{equation}
 that is, the ratio between $\gamma(p^+\!)^{-1}$ times the coefficient of the leading vanishing term of $g_{\mathrm{reg}}^+$ and the coefficient of the leading divergent term of $g_{\mathrm{sing}}^+$ is indexed by the real extension parameter $\beta$.
\end{rem}

\section{Homogeneous problem $\widetilde{S}u=0$}\label{sec:kernel}

In this Section we identify the dimension and the basis of the subspace $\ker S^*$, and prove Proposition \ref{prop:kerS*}. One has to solve the homogeneous differential equation $\widetilde{S}u=0$, where $\widetilde{S}$ is the differential operator \eqref{eq:tildeS} and the function $r\mapsto u(r)=\begin{pmatrix} u^+\!(r) \\ u^-\!(r) \end{pmatrix}$ on $\mathbb{R}^+$ is the spinorial unknown.
The needed ODE technique is classical and we include it concisely for completeness. Observe, however, that for the application of von Neumann's theory of self-adjoint extension one has to solve the ODE problem $\widetilde{S}u=z u$ for non-real $z\in\mathbb{C}$, say, $z=\pm\ii$, which requires a somewhat more extended discussion -- see, e.g., \cite[Section 3]{Voronov-Gitman-Tyutin-TMP2007} or \cite[Sections 4-6]{Hogreve-2013_JPhysA}.

Upon transforming the unknown $u$ into $\varphi$, where
\begin{equation}\label{eq:u-varphi}
 \varphi(r)\;:=\;{\textstyle\frac{1}{2}}(\mathbf{A}u)({\textstyle\frac{r}{2}})\,e^{r/2}\,,\qquad \mathbf{A}:=\frac{1}{\sqrt{2}}\begin{pmatrix} 1 & 1 \\ 1 & -1 \end{pmatrix},
\end{equation}
the differential system $\widetilde{S}u=0$ takes the form
\begin{equation}\label{eq:diff_varphi}
 \begin{cases}
  \,(\varphi^+\!)'\,=\,\varphi^+\!-\frac{1-\nu}{r}\,\varphi^- \\
  \,(\varphi^-\!)'\,=\,-\frac{1+\nu}{r}\,\varphi^+\,.
 \end{cases}
\end{equation}
Therefore, $\varphi^-$ is a solution to
\begin{equation}
 r(\varphi^-\!)''+(1-r)(\varphi^-\!)'-\textstyle{\frac{\nu^2-1}{r}}\,\varphi^-\!\;=\;0\,,
\end{equation}
equivalently,
\begin{equation}\label{eq:xi-phi}
 \xi(r)\,:=\,r^B\varphi^-\!(r)
\end{equation}
is a solution to
\begin{equation}\label{eq:ODE_hyper}
 r\xi''+(1-2B-r)\xi'+B\,\xi\;=\;0\,.
\end{equation}
The second order ODE \eqref{eq:ODE_hyper} is the confluent hypergeometric equation \cite[(13.1.1) and (13.1.11)]{Abramowitz-Stegun-1964}, two linearly independent solutions of which are the confluent hypergeometric functions of first and second kind, that is, respectively, the Kummer function $M_{a,b}(r)$ \cite[(13.1.2)]{Abramowitz-Stegun-1964} and the Tricomi function $U_{a,b}(r)$ \cite[(13.1.3)]{Abramowitz-Stegun-1964}, with $a=-B$ and $b=1-2B$.

The solutions $\xi_0(r):=M_{-B,1-2B}(r)$ and $\xi_\infty(r):=U_{-B,1-2B}(r)$ to \eqref{eq:ODE_hyper} determine, via \eqref{eq:xi-phi} and the second of \eqref{eq:diff_varphi}, two linearly independent solutions $\varphi_0=\begin{pmatrix}\varphi_0^+\! \\ \varphi_0^-\! \end{pmatrix}$ and $\varphi_\infty=\begin{pmatrix}\varphi_\infty^+\! \\ \varphi_\infty^-\! \end{pmatrix}$ to \eqref{eq:diff_varphi}. Using the properties
\[
M_{a,b}'(r)\,=\,\frac{a}{b}\,M_{a+1,b+1}(r)\,,\qquad U_{a,b}'(r)\,=\,-a\,U_{a+1,b+1}(r)
\]
(\cite[(13.4.8) and (13.4.21)]{Abramowitz-Stegun-1964}), and the inverse transformation of \eqref{eq:u-varphi}, that is, $u(r)=2\,e^{-r/2}(\mathbf{A}^{-1}\varphi)(2r)$, where $\mathbf{A}^{-1}=\mathbf{A}$, yields the following two linearly independent solutions to the original problem $\widetilde{S}u=0$:
\begin{equation}\label{eq:solutions_u0_uinf}
\begin{split}
\!\!\!\!\!\!u_0(r)&:=\frac{1}{e^{r}r^{B}}\!\begin{pmatrix}
 \textstyle{\frac{1+\nu+B}{1+\nu}}\,M_{-B,1-2B}(2r)+\textstyle{\frac{2rB}{(1+\nu)(1-2B)}}\,M_{1-B,2-2B}(2r) \\
 \textstyle{-\frac{1+\nu-B}{1+\nu}}\,M_{-B,1-2B}(2r)+\textstyle{\frac{2rB}{(1+\nu)(1-2B)}}\,M_{1-B,2-2B}(2r)
                               \end{pmatrix} \\
 \!\!\!\!\!\!\!\!u_\infty(r)&:=\frac{1}{e^{r}r^{B}}\!\begin{pmatrix}
 \textstyle{\frac{1+\nu+B}{1+\nu}}\,U_{-B,1-2B}(2r)-\textstyle{\frac{2rB}{1+\nu}}\,U_{1-B,2-2B}(2r) \\
 \textstyle{-\frac{1+\nu-B}{1+\nu}}\,U_{-B,1-2B}(2r)-\textstyle{\frac{2rB}{1+\nu}}\,U_{1-B,2-2B}(2r)
                               \end{pmatrix}
\end{split}
\end{equation}
(in fact, an irrelevant common pre-factor $2^{-B}$ has been neglected). Both $u_0$ and $u_\infty$ are real-valued and smooth on $\mathbb{R}^+$.

Because of the asymptotics \cite[(13.1.2), (13.5.1), and (13.5.5)]{Abramowitz-Stegun-1964}
\begin{equation}
 \begin{split}
  M_{a,b}(r)\;&=\;\frac{\,e^r\,r^{a-b}\,}{\Gamma(a)}\,(1+O(r^{-1}))\qquad\textrm{as }\,r\to+\infty \\
  M_{a,b}(r)\;&=\;1+O(r)\qquad\qquad\;\;\textrm{as }\,r\downarrow 0\quad\textrm{ and }-\!b\notin\mathbb{N}
 \end{split}
\end{equation}
and \cite[(13.1.2), (13.1.3), (13.5.2), (13.5.8), and (13.5.10)]{Abramowitz-Stegun-1964}
\begin{equation}
 \begin{split}
  U_{a,b}(r)\;&=\;r^{-a}(1+O(r^{-1}))\qquad\qquad\qquad\qquad\qquad\,\textrm{as }\,r\to+\infty \\
  U_{a,b}(r)\;&=\;\frac{\Gamma(1-b)}{\Gamma(1+a-b)}+\frac{\Gamma(b-1)}{\Gamma(a)}\,r^{1-b}+O(r)\quad
  \begin{array}{l}
   \textrm{as }\,r\downarrow 0 \\
   \textrm{and }b\in(0,1)
  \end{array} \\
  U_{a,b}(r)\;&=\;\frac{\Gamma(b-1)}{\Gamma(a)}\,r^{-(b-1)}+O(1)\qquad\qquad\qquad\;\;
  \begin{array}{l}
   \textrm{as }\,r\downarrow 0 \\
   \textrm{and }b\in(1,2)\,,
  \end{array}
 \end{split}
\end{equation}
one deduces that both $u_0$ and $u_\infty$ are square-integrable around $r=0$, whereas only $u_\infty$ is square-integrable at infinity, and moreover
\begin{equation}\label{eq:asymptotics_for_u0}
 \begin{split}
  u_0(r)\;&=\;\begin{pmatrix}
               \frac{1+\nu+B}{1+\nu} \\ -\frac{1+\nu-B}{1+\nu}
              \end{pmatrix} r^{-B} +O(r^{1-B})\qquad\qquad\qquad\,\textrm{as } r\downarrow 0\\
  u_0(r)\;&=\;-\textstyle{\frac{2^B(1-2B)}{\Gamma(-B)(1+\nu)}}\begin{pmatrix}
               1 \\ 1
              \end{pmatrix}r^Be^r (1+O(r^{-1}))\qquad\qquad\textrm{as } r\to +\infty\,,
 \end{split}
\end{equation}
and
\begin{equation}\label{eq:asymptotics_for_uinfty}
 \begin{split}
  u_\infty(r)\;&=\;\textstyle{\frac{\Gamma(2B)}{\Gamma(B)}}\begin{pmatrix}
               \frac{1+\nu-B}{1+\nu} \\ -\frac{1+\nu+B}{1+\nu}
              \end{pmatrix} r^{-B} +\begin{pmatrix} q^+\! \\ q^-\! \end{pmatrix} r^B +O(r^{1-B})\quad\:\textrm{as } r\downarrow 0\\
  u_\infty(r)\;&=\;2^B\begin{pmatrix}
               1 \\ -1
              \end{pmatrix}r^{-B}e^{-r}
              (1+O(r^{-1}))\qquad\textrm{as } r\to +\infty\,,
 \end{split}
\end{equation}
where
\begin{equation}\label{eq:def_qpm}
 q^\pm\;:=\;\textstyle\frac{4^B(-B\pm(1+\nu))\Gamma(-2B)}{(1+\nu)\Gamma(-B)}\,.
\end{equation}
Observe that $q^\pm\neq 0$.

Therefore, there is only a \emph{one}-dimensional space of solutions to $\widetilde{S}u=0$ which are square integrable, and hence, owing to \eqref{DSclosureDS*}, $\ker S^*$ is one-dimensional. For convenience, let us choose as the spanning vector the function $\Phi:=u_\infty$. Then \eqref{eq:asymptotics_for_uinfty} implies \eqref{eq:Phi_asymptotics} and Proposition \ref{prop:kerS*} is proved.

\section{Distinguished extension $S_D$}\label{sec:distinguished}

In this Section we qualify the distinguished extension $S_D$ of the operator $S$, and prove Proposition \ref{prop:SD}. 

When comparing the approach based on von Neumann's theory, as developed, e.g., in \cite{Hogreve-2013_JPhysA}, with the present one based on the Kre{\u\i}n-Vi\v{s}ik-Birman theory, to solve the homogeneous problem $S^*u=0$ (in the KVB strategy) or to solve the deficiency space problem $S^*u=\pm\ii u$ (in the von Neumann strategy) are two essentially analogous versions of the same step, from the ODE viewpoint. In contrast, the qualification of $S_D$ (in view of Theorem \ref{thm:VB-representaton-theorem_Tversion}, strictly speaking one only needs to qualify the action of $S_D^{-1}$ on $\ker S^*$) is a specific step of the KVB strategy, and it boils down to solving the ODE problem $\widetilde{S}f=g$ for given $g$. Along this line, we adapt to our case the analysis done in \cite{B-Derezinski-G-AHP2011} for homogeneous Schr\"{o}dinger operators on half-line.

In order to set up the problem conveniently, let us first replace the pair $(u_0,u_\infty)$ of linearly independent solutions \eqref{eq:solutions_u0_uinf} to $\widetilde{S}u=0$ to the new pair $(v_0,v_\infty)$ given by
\begin{equation}\label{eq:v0_vinf}
 \begin{split}
  v_0\;&:=u_\infty-{\textstyle \frac{\Gamma(2B)}{\Gamma(B)}}\,u_0 \\
  v_\infty\;&:=\;u_\infty \,.
 \end{split}
\end{equation}
This preserves the linear independence of $v_0$ and $v_\infty$ with the virtue of having two solutions with different power-law in the asymptotics as $r\downarrow 0$: from \eqref{eq:v0_vinf} and \eqref{eq:asymptotics_for_u0}-\eqref{eq:asymptotics_for_uinfty} we find 
\begin{equation}\label{eq:asymptotics_for_v0vinf}
\begin{array}{l}
 \begin{split}
 v_0(r)\;&=\;\begin{pmatrix} q^+\! \\ q^-\! \end{pmatrix} r^B+O(r^{1-B})\\
 v_\infty(r)\;&=\;\textstyle{\frac{\Gamma(2B)}{\Gamma(B)}}\begin{pmatrix}
               \frac{1+\nu+B}{1+\nu} \\ -\frac{1+\nu-B}{1+\nu}
              \end{pmatrix} r^{-B} +O(r^B)
 \end{split}
\end{array}\qquad\textrm{as } r\downarrow 0\,,
\end{equation}
where $q^\pm$ is given by \eqref{eq:def_qpm}. At large distances, $v_0$ and $v_\infty$ have exponential 
asymptotics as $u_0$ and $u_\infty,$ namely
\begin{equation}\label{eq:asymptotics_for_v0vinf_large_distances}
\begin{array}{l}
 \begin{split}
 v_0(r)\;&=\;-\textstyle{\frac{1}{2} \,\frac{2^{B} B}{(1+\nu) \cos(B \pi)}} \begin{pmatrix}
               1 \\ 1
              \end{pmatrix}r^B e^r(1+O(r^{-1}))\\
 v_\infty(r)\;&=\;2^B\begin{pmatrix}
               1 \\ -1
              \end{pmatrix}r^{-B} e^{-r} 
              (1+O(r^{-1}))
 \end{split}
\end{array}\;\textrm{as } r\to+\infty\,.
\end{equation}

We then proceed with standard ODE arguments. With respect to the fundamental system $(v_0,v_\infty)$, the general solution to the inhomogeneous problem $\widetilde{S}f=g$ has the form
\begin{equation}\label{eq:ODE-decomp}
 f\;=\;A_0v_0+A_\infty v_\infty+f_\mathrm{part}\,,
\end{equation}
where $A_0$ and $A_\infty$ run over $\mathbb{C}$ and $f_\mathrm{part}$ is a particular solution, namely, $\widetilde{S}f_\mathrm{part}=g$. Let us determine it through the variation of constants \cite[Section 2.4]{Wasow_asympt_expansions}.

First we re-write $\widetilde{S}f=g$ in normal form as
\begin{equation}\label{eq:normal_form}
 y'+\mathbf{V}(r)y\,=\,g\,,\qquad y\,:=\,\mathbf{E}f\,,
\end{equation}
where
\begin{equation}\label{eq:normal_form_potential}
 \mathbf{V}(r)\,:=\frac{1}{r}
 \begin{pmatrix}
 -1 & \nu \\
 -\nu & 1 
 \end{pmatrix}+
 \begin{pmatrix}
  0 & 1 \\ 1 & 0
 \end{pmatrix},
 \qquad \mathbf{E}\,:=\begin{pmatrix}
                                                                                 0 & -1 \\1 & 0
                                                                                \end{pmatrix}\,.
\end{equation}
We also introduce the Wronskian
\begin{equation}
 \mathbb{R}^+\ni r\mapsto W_r(v_0,v_\infty)\;:=\;\det\begin{pmatrix}
                             v_0^+(r) & v_\infty^+(r) \\
                             v_0^-(r) & v_\infty^-(r)
                            \end{pmatrix}.
\end{equation}
This is precisely the Wronskian $ W_r(\mathbf{E}v_0,\mathbf{E}v_\infty)$ of two fundamental solutions of the problem written in normal form, because $\det\mathbf{E}=1$ and hence $ W_r(\mathbf{E}v_0,\mathbf{E}v_\infty)=W_r(v_0,v_\infty)$. Moreover, since $\mathbf{V}(r)$ is traceless for any $r\in\mathbb{R}^+$, Liouville's theorem implies that $W_r(-\mathbf{E}v_0,-\mathbf{E}v_\infty)$ is constant, and so is also $W_r(v_0,v_\infty)$. Therefore,
\begin{equation}\label{eq:computation_of_wronskian}
\begin{split}
 W_r(v_0,v_\infty)\;&=\;\lim_{r\downarrow 0}W_r(v_0,v_\infty)\;=\; \textstyle{\frac{4^B B}{(1+\nu)\cos(B \pi)}} \;=:\;W_0^\infty\,.
\end{split}
\end{equation}
The limit in \eqref{eq:computation_of_wronskian} above follows straightforwardly from the asymptotics \eqref{eq:asymptotics_for_v0vinf} and from the expression \eqref{eq:def_qpm} for $q^\pm$. Clearly, $W_0^\infty\neq 0$. Then a standard application of the method of variation of constants for the differential problem \eqref{eq:normal_form} and the further transformation $f=\mathbf{E}y$ yields eventually
\begin{equation}
 f_\mathrm{part}(r)\;=\;\int_0^{+\infty}\!\!G(r,\rho)\,g(\rho)\,\ud \rho\,,
\end{equation}
where
\begin{equation}\label{eq:Green}
 G(r,\rho)\;:=\;
 \begin{cases}
  \frac{1}{\:W_0^\infty}\begin{pmatrix}
                    v^+_\infty(r) v^+_0(\rho) & v^+_\infty(r) v^-_0(\rho) \\
                    v^-_\infty(r) v^+_0(\rho) & v^-_\infty(r) v^-_0(\rho)
                   \end{pmatrix} & \textrm{if } 0<\rho<r \\
                   \\
  \frac{1}{\:W_0^\infty}\begin{pmatrix}
                    v^+_0(r) v^+_\infty(\rho) & v^+_0(r) v^-_\infty(\rho) \\
                    v^-_0(r) v^+_\infty(\rho) & v^-_0(r) v^-_\infty(\rho)
                   \end{pmatrix} & \textrm{if } 0<r<\rho\,.
 \end{cases}
\end{equation}

Next, we observe the following.

\begin{lem}\label{lem:RG_bounded}
 The integral operator $R_G$ on $L^2(\mathbb{R}^+,\mathbb{C},\ud r)$ with kernel $G(r,\rho)$ given by \eqref{eq:Green} is bounded and self-adjoint.
\end{lem}

\begin{proof}
 For each $r,\rho\in\mathbb{R}^+$, $G(r,\rho)$ is the sum of the four terms
 \begin{equation}\label{eq:four_term_G}
 \begin{split}
  G^{++}(r,\rho)\;&:=\;G(r,\rho)\,\mathbf{1}_{(1,+\infty)}(r)\,\mathbf{1}_{(1,+\infty)}(\rho) \\
  G^{+-}(r,\rho)\;&:=\;G(r,\rho)\,\mathbf{1}_{(1,+\infty)}(r)\,\mathbf{1}_{(0,1)}(\rho) \\
  G^{-+}(r,\rho)\;&:=\;G(r,\rho)\,\mathbf{1}_{(0,1)}(r)\,\mathbf{1}_{(1,+\infty)}(\rho) \\
  G^{--}(r,\rho)\;&:=\;G(r,\rho)\,\mathbf{1}_{(0,1)}(r)\,\mathbf{1}_{(0,1)}(\rho)\,, 
 \end{split}
 \end{equation}
where $\mathbf{1}_J$ denotes the characteristic function of the interval $J\subset\mathbb{R}^+$, and correspondingly $R_G$ splits into the sum of four integral operators with kernel given by \eqref{eq:four_term_G}.

Now, for each entry of $G^{LM}(r,\rho)$, with $L,M\in\{+,-\}$, a point-wise estimate in $(r,\rho)$ can be derived from the short and large distance asymptotics for $v_0$ and $v_\infty$. For example, the entry $G^{++}_{11}(r,\rho)$ in the first row and first column of $G^{++}(r,\rho)$ is controlled as
\[
 \begin{split}
  |v^+_\infty(r) \,v^+_0(\rho)\,\mathbf{1}_{(1,+\infty)}(r)\,\mathbf{1}_{(1,+\infty)}(\rho)|\;&\lesssim\; e^{-r} \,e^\rho \, (\rho/r)^B\qquad\textrm{if }\,0<\rho<r \\
  |v^+_0(r)\, v^+_\infty(\rho)\,\mathbf{1}_{(1,+\infty)}(r)\,\mathbf{1}_{(1,+\infty)}(\rho)|\;&\lesssim\; e^r\,e^{-\rho}\, (r/\rho)^B \qquad\textrm{if }\,0<r<\rho\,,
 \end{split}
\]
because $v_0$ diverges exponentially and $v_\infty$ vanishes exponentially as $r\to+\infty$, \eqref{eq:asymptotics_for_v0vinf_large_distances}; thus,
\[
 |G^{++}_{11}(r,\rho)|\;\lesssim\;e^{-|r-\rho|}\,(\rho r)^B \,.
\]
In fact, the asymptotics for $v_0$ and $v_\infty$ are the same for both components, so we can also conclude that 
\[
 \|G^{++}(r,\rho)\|_{M_2(\mathbb{C})}\;\lesssim\;e^{-|r-\rho|}\, (\rho r)^B\, ,
\]
where $\|\cdot\|_{M_2(\mathbb{C})}$ denotes the matrix norm. The estimate of the other kernels is perfectly analogous, and we find
\begin{equation}\label{eq:matrix_estimates}
 \begin{split}
  \|G^{++}(r,\rho)\|_{M_2(\mathbb{C})}\;&\lesssim\;(r \rho)^B e^{-|r-\rho|}\,\mathbf{1}_{(1,+\infty)}(r)\,\mathbf{1}_{(1,+\infty)}(\rho) \\
  \|G^{+-}(r,\rho)\|_{M_2(\mathbb{C})}\;&\lesssim\;r^Be^{-\rho}\,\mathbf{1}_{(1,+\infty)}(r)\,\mathbf{1}_{(0,1)}(\rho) \\
  \|G^{-+}(r,\rho)\|_{M_2(\mathbb{C})}\;&\lesssim\;e^{-r}\rho^B\,\mathbf{1}_{(0,1)}(r)\,\mathbf{1}_{(1,+\infty)}(\rho) \\
  \|G^{--}(r,\rho)\|_{M_2(\mathbb{C})}\;&\lesssim\;(r^B\rho^{-B}+r^{-B}\rho^B)\,\mathbf{1}_{(0,1)}(r)\,\mathbf{1}_{(0,1)}(\rho)\,.
 \end{split}
\end{equation}

The last three estimates in \eqref{eq:matrix_estimates} show at once that the kernels $G^{+-}(r,\rho)$, $G^{-+}(r,\rho)$, and $G^{--}(r,\rho)$ are in $L^2(\mathbb{R}^+\times\mathbb{R}^+,M_2(\mathbb{C}),\ud r\,\ud\rho)$ and therefore the corresponding integral operators are Hilbert-Schmidt operators, hence bounded, on $L^2(\mathbb{R}^+,\mathbb{C}^2,\ud r)$. The first estimate in \eqref{eq:matrix_estimates} allows to conclude, by an obvious Schur test, that also the integral operator with kernel $G^{++}(r,\rho)$ is bounded on $L^2(\mathbb{R}^+,\mathbb{C}^2,\ud r)$. This proves the overall boundedness of $R_G$.

The self-adjointness of $R_G$ is clear from \eqref{eq:Green}: the adjoint $R_G^*$ of $R_G$ has kernel $\overline{G(\rho,r)^T}$, but $G$ is real-valued and $G(\rho,r)=G(r,\rho)$, thus proving that $R_G^*=R_G$.
\end{proof}

The integral operator $R_G$ has a relevant mapping property that is more directly read out from the following alternative representation.
If $f_\mathrm{part}=R_G\,g$, then
\begin{equation}\label{eq:fpart_alternative_repr}
 f_\mathrm{part}(r)\;=\;\Theta_\infty^{(g)}(r)\,v_0(r)+\Theta_0^{(g)}(r)\,v_\infty(r)\,,
\end{equation}
where
\begin{equation}\label{eq:fpart_alternative_repr-Thetas}
 \begin{split}
  \Theta_0^{(g)}(r)\;&:=\;\frac{1}{\:W_0^\infty}\int_0^r\langle{\,\overline{v_0(\rho)}}\,,\,g(\rho)\,\rangle_{\mathbb{C}^2}\,\ud \rho \\
    \Theta_\infty^{(g)}(r)\;&:=\;\frac{1}{\:W_0^\infty}\int_r^{+\infty}\!\langle{\,\overline{v_\infty(\rho)}}\,,\,g(\rho)\,\rangle_{\mathbb{C}^2}\,\ud \rho
 \end{split}
\end{equation}
and $W_0^\infty$ is the constant computed in \eqref{eq:computation_of_wronskian}. Indeed, from \eqref{eq:Green},
\[
\begin{split}
 f_\mathrm{part}(r)\;&=\;\int_0^{+\infty}\!G(r,\rho)\,g(\rho)\,\ud \rho \\
 &=\;\frac{1}{W_0^{\infty}}\begin{pmatrix} 
 v_\infty^+(r) \\ v_\infty^-(r)
 \end{pmatrix}\int_0^r\big(v_0^+(\rho)g^+(\rho)+v_0^-(\rho)g^-(\rho)\big)\,\ud \rho \\
 &\qquad +\frac{1}{W_0^{\infty}}\begin{pmatrix} 
 v_0^+(r) \\ v_0^-(r)
 \end{pmatrix}\int_r^{+\infty}\!\!\big(v_\infty^+(\rho)g^+(\rho)+v_\infty^-(\rho)g^-(\rho)\big)\,\ud \rho\,,
\end{split}
\]
that is, \eqref{eq:fpart_alternative_repr}.

\begin{lem}\label{RG_maps_in_energyform}
For every $g\in L^2(\mathbb{R}^+,\mathbb{C}^2)$ one has
\begin{equation}\label{eq:Rgg_finite_pot_energy}
 \int_0^{+\infty}\frac{\|(R_G\,g)(r)\|^2_{\mathbb{C}^2}}{r}\,\ud r\;<\;+\infty\,,
\end{equation}
i.e.,
\begin{equation}
 \mathrm{ran}\,R_G\;\subset\;\mathcal{D}[r^{-1}]\,.
\end{equation}
\end{lem}

\begin{proof}
It suffices to prove the finiteness of the integral in \eqref{eq:Rgg_finite_pot_energy} only for $r\in(0,1)$, since $\int_1^{+\infty}r^{-1}\|(R_G\,g)(r)\|^2_{\mathbb{C}^2}\,\ud r\leqslant\|R_G\|^{2}\|g\|^2_{L^2(\mathbb{R}^+,\mathbb{C}^+)}$\,. Let us represent $f=R_G\,g\in\mathrm{ran}\,R_G$ as in \eqref{eq:fpart_alternative_repr}-\eqref{eq:fpart_alternative_repr-Thetas}. For $r\in(0,1)$ one has
\[
 \begin{split}
   |\Theta_0^{(g)}(r)|\;&\leqslant\;|W_0^\infty|^{-1}\,\|v_0\,\mathbf{1}_{(0,r)}\|_{L^2(\mathbb{R}^+,\mathbb{C}^2)}\,\|g\|_{L^2(\mathbb{R}^+,\mathbb{C}^2)}\;\leqslant\;C_{g,\nu}\,r^{B+\frac{1}{2}}  \\
   |\Theta_\infty^{(g)}(r)|\;&\leqslant\;|W_0^\infty|^{-1}\,\|v_\infty\,\mathbf{1}_{(r,\infty)}\|_{L^2(\mathbb{R}^+,\mathbb{C}^2)}\,\|g\|_{L^2(\mathbb{R}^+,\mathbb{C}^2)}\;\leqslant\;C_{g,\nu}
 \end{split}
\]
for some constant $C_{g,\nu}>0$ depending on $g$ and $\nu$ only, having used the short distance asymptotics \eqref{eq:asymptotics_for_v0vinf} for $v_0$ and $v_\infty$.
%
%
%
%
%
%
Combining now the above bounds again with \eqref{eq:asymptotics_for_v0vinf} we see that on the interval $(0,1)$ the functions $r\mapsto \Theta_0^{(g)}(r)\,v_\infty(r)$ and $r\mapsto \Theta_\infty^{(g)}(r)\,v_0(r)$ are continuous and vanish when $r\to 0$, respectively, as $r^{1/2}$ and $r^B$, which makes the integral $\int_0^{1}r^{-1}\|(R_G\,g)(r)\|^2_{\mathbb{C}^2}\,\ud r$ finite.
%
%
%
\end{proof}

Combining Lemmas  \ref{lem:RG_bounded} and \ref{RG_maps_in_energyform} together, we are now in the condition to prove Proposition \ref{prop:SD}.

\medskip

\begin{proof}[Proof of Proposition \ref{prop:SD}]~

(i) and (ii). The integral operator $R_G$ on $L^2(\mathbb{R}^+,\mathbb{C})$ with kernel given by \eqref{eq:Green} is bounded and self-adjoint owing to Lemma \ref{lem:RG_bounded}, and by construction satisfies $\widetilde{S}\,R_G\,g=g$ $\forall g\in L^2(\mathbb{R}^+,\mathbb{C})$. Therefore, there is one self-adjoint extension $\mathscr{S}$ of $S_\mathrm{min}=\overline{S}$ such that $\mathscr{S}R_G\,g=g$ $\forall g\in L^2(\mathbb{R}^+,\mathbb{C})$, whence, by self-adjointness, also $R_G\mathscr{S}h=h$ $\forall h\in\mathcal{D}(\mathscr{S})$. Thus, $R_G=\mathscr{S}^{-1}$ for some invertible self-adjoint realisation $\mathscr{S}$ of $S$. Because of Lemma \ref{RG_maps_in_energyform}, the space $\mathcal{D}(\mathscr{S})=\mathrm{ran}\,R_G$ is contained in the potential energy form domain $\mathcal{D}[r^{-1}]$: owing to Theorem \ref{thm:recap}(ii) then $\mathscr{S}$ must be the reduction to the subspace $\cH_{\frac{1}{2},1}$ of the distinguished self-adjoint extension of the Dirac-Coulomb operator $H$: we shall denote it with $S_D$. As such, $S_D$ is the unique self-adjoint realisation of $S$ satisfying the property \eqref{eq:SD_uniqueness_properties}, it is invertible, and its kernel is precisely given by \eqref{eq:Green}.

(iii) The decompositions \eqref{eq:decomp_DSD} and \eqref{eq:Krein_decomp_formula} are canonical, once a self-adjoint extension of $S$ is given with everywhere defined and bounded inverse: see, for instance, \cite[(2.4) and (2.5)]{GMO-KVB2017}.

(iv) From the previous discussion, $\Phi=u_\infty=v_\infty$ and $S_D^{-1}\Phi=R_Gv_\infty$. A closed expression for the latter function is given by \eqref{eq:fpart_alternative_repr} above, which now reads
\[
 S_D^{-1}\Phi\;=\;\Theta_\infty^{(v_\infty)}(r)\,v_0(r)+\Theta_0^{(v_\infty)}(r)\,v_\infty(r)\,.
\]
From \eqref{eq:fpart_alternative_repr-Thetas} and \eqref{eq:asymptotics_for_v0vinf} we deduce
\[
 \begin{split}
  |\Theta_0^{(v_\infty)}(r)|\;&\leqslant\;|W_0^\infty|^{-1}\int_0^r\big|\langle{\,\overline{v_0(\rho)}}\,,\,v_\infty(\rho)\,\rangle_{\mathbb{C}^2}\big|\,\ud \rho \\
  &\lesssim\;\int_0^r(\rho^B+O(\rho^{1-B}))(\rho^{-B}+O(\rho^{B}))\,\ud \rho \\
  &=\;r+o(r)\quad\textrm{as }\;r\downarrow 0
 \end{split}
\]
and 
 \[
 \begin{split}
  \Theta_\infty^{(v_\infty)}(r)\;&=\;\frac{1}{\,W_0^\infty}\int_r^{+\infty}\langle{\,\overline{v_\infty(\rho)}}\,,\,v_\infty(\rho)\,\rangle_{\mathbb{C}^2}\,\ud \rho \\
  &=\;\frac{1}{\,W_0^\infty}\,\|v_\infty\|_{L^2(\mathbb{R}^+,\mathbb{C}^2)}^2(1+o(1))\quad\textrm{as }\;r\downarrow 0\,.
 \end{split}
\]
Therefore, using again the short distance asymptotics \eqref{eq:asymptotics_for_v0vinf},
\[\tag{*}
 \begin{split}
  (S_D^{-1}\Phi)(r)\;=\;\frac{\|v_\infty\|_{L^2(\mathbb{R}^+,\mathbb{C}^2)}^2}{\,W_0^\infty}\,
  \begin{pmatrix}
   q^+\! \\ q^-\!
  \end{pmatrix}
  r^{B}+o(r^{B})
 \end{split}
\]
where $q^{\pm}$ is given by \eqref{eq:def_qpm}.
Upon setting
\begin{equation}\label{ed:defppm}
 p^{\pm}\;:=\; q^{\pm}\,(W_0^\infty)^{-1}\|v_\infty\|_{L^2(\mathbb{R}^+,\mathbb{C}^2)}^2 
\end{equation}
we then obtain the leading term of \eqref{SDPhi_vanishing}. The remainder is in fact smaller than $o(r^B)$. This can be seen by comparing the above asymptotics for $S_D^{-1}\Phi$ with the expansion \eqref{eq:decomp_a0ainf_b0binf} established in the next Section (which is valid because $S_D^{-1}\Phi\in\mathcal{D}(S_D)\subset\mathcal{D}(S^*)$), namely
\[
 S_D^{-1}\Phi\;=\;a_0^{(S_D^{-1}\Phi)}\,v_0+a_\infty^{(S_D^{-1}\Phi)}\,v_\infty+b_\infty^{(S_D^{-1}\Phi)}\,v_0+ b_0^{(S_D^{-1}\Phi)}\,v_\infty\,.
\]
For the latter, we have the asymptotics
\[\tag{**}
\begin{split}
 \!\!(S_D^{-1}\Phi)(r)\;&=\;\mathbf{c}_0 \,a_0^{(S_D^{-1}\Phi)}\,(r^B+O(r^{1-B}))+\mathbf{c}_\infty\, a_\infty^{(S_D^{-1}\Phi)}\,(r^{-B}+O(r^B)) \\
 &\qquad +o(r^{1/2})\quad\textrm{as }\;r\downarrow 0\,.
\end{split}
\]
as follows from \eqref{eq:asymptotics_for_v0vinf} and \eqref{eq:bvbg_asympt} for some non-zero constants $\mathbf{c}_0,\mathbf{c}_\infty\in\mathbb{C}^2$. In order for (*) and (**) to be compatible, necessarily $a_\infty^{(S_D^{-1}\Phi)}=0$. This implies that after the leading order $r^B$ there comes a remainder $o(r^{1/2})$, thus completing the proof of \eqref{SDPhi_vanishing}.
\end{proof}

\section{Operator closure $\overline{S}$}\label{sec:closure}

This Section is devoted to the proof of Proposition \ref{prop:Sclosure}. In fact we will prove a stronger result of characterisation of $\mathcal{D}(\overline{S})$, namely Proposition \ref{prop:characterisationDS} below, from which Proposition \ref{prop:Sclosure} follows as a corollary.

Let us start with another useful representation of $\mathcal{D}(S^*)$. It is analogous to the operator-theoretic decomposition \eqref{eq:decomp_DSD}, but its formulation (and proof) is more in the ODE spirit.

\begin{lem}\label{lem:decomp_a0ainf_b0binf}
 For each $g\in\mathcal{D}(S^*)$ there exist, uniquely determined, constants $a_0^{(g)},a_\infty^{(g)}\in\mathbb{C}$ and functions
 \begin{equation}\label{eq:b0binf}
  \begin{split}
   b_0^{(g)}(r)\;&:=\;\frac{1}{W_0^\infty}\int_0^r\langle\, \overline{v_0(\rho)}\,,\,(S^*\!g)(\rho)\,\rangle_{\mathbb{C}^2}\,\ud\rho \\
   b_\infty^{(g)}(r)\;&:=\;-\frac{1}{W_0^\infty}\int_0^r\langle \,\overline{v_\infty(\rho)}\,,\,(S^*\!g)(\rho)\,\rangle_{\mathbb{C}^2}\,\ud\rho
  \end{split}
 \end{equation}
 on $\mathbb{R}^+$ such that
 \begin{equation}\label{eq:decomp_a0ainf_b0binf}
  g\;=\;a_0^{(g)}\,v_0+a_\infty^{(g)}\,v_\infty+b_\infty^{(g)}\,v_0+ b_0^{(g)}\,v_\infty\,,
 \end{equation}
 where $v_0$ and $v_\infty$ are the two linearly independent solutions \eqref{eq:v0_vinf} to the homogeneous  problem $\widetilde{S}v=0$ (recall that they are real and smooth on $\mathbb{R}^+$) and $W_0^\infty$ is the constant computed in \eqref{eq:computation_of_wronskian}. Moreover, both $ b_0^{(g)}(r)$ and $b_\infty^{(g)}(r)$ vanish as $r\downarrow 0$, and
 \begin{equation}\label{eq:bvbg_asympt}
  b_\infty^{(g)}(r)\,v_0(r)+ b_0^{(g)}(r)\,v_\infty(r)\;=\;o(r^{1/2})\qquad\textrm{as }\;r\downarrow 0\,.
 \end{equation}
\end{lem}


\begin{proof}
Let $h:=S^*g=\widetilde{S}g$. Then, as already argued in \eqref{eq:ODE-decomp} and \eqref{eq:fpart_alternative_repr}-\eqref{eq:fpart_alternative_repr-Thetas}, $g$ is expressed in terms of $h$ as
\[\tag{*}
 g\;=\;A_0\,v_0+A_\infty\,v_\infty+\Theta^{(h)}_\infty\,v_0+ \Theta_0^{(h)}\,v_\infty
\]
for some $A_0,A_\infty\in\mathbb{C}$ that are now uniquely identified by $g$. From \eqref{eq:fpart_alternative_repr-Thetas} and \eqref{eq:b0binf} we see that
\[
 \begin{split}
  \Theta_\infty^{(h)}(r)\;&=\;b_\infty^{(g)}(r) \\
  \Theta_0^{(h)}(r)\;&=\;-\frac{1}{\:W_0^\infty}\int_0^r\langle{\,\overline{v_0(\rho)}}\,,\,(S^*\!g)(\rho)\,\rangle_{\mathbb{C}^2}\,\ud \rho \\
  &= \;b_0^{(g)}(r)-\frac{1}{\:W_0^\infty}\int_0^{+\infty}\!\!\langle{\,\overline{v_0(\rho)}}\,,\,(S^*\!g)(\rho)\,\rangle_{\mathbb{C}^2}\,\ud \rho\,.
 \end{split}
\]
Then (*) implies \eqref{eq:decomp_a0ainf_b0binf} at once, upon setting 
\[
 \begin{split}
  a_0^{(g)}\;&:=\;A_0-\frac{1}{\:W_0^\infty}\int_0^{+\infty}\!\!\langle{\,\overline{v_0(\rho)}}\,,\,(S^*\!g)(\rho)\,\rangle_{\mathbb{C}^2}\,\ud \rho \\
  a_\infty^{(g)}\;&:=\;A_\infty
 \end{split}
\]
Observe that the constant added above to $A_0$ is finite and bounded by $|W_0^\infty|^{-1}\|v_0\|_{L^2(\mathbb{R}^+,\mathbb{C}^2)}\|S^*\!g\|_{L^2(\mathbb{R}^+,\mathbb{C}^2)}$.
As for the proof of \eqref{eq:bvbg_asympt}, by means of the short distance asymptotics \eqref{eq:asymptotics_for_v0vinf} for $v_0$ and $v_\infty$ we find
\[
 \begin{split}
  |b_\infty^{(g)}(r)\,v_0(r)|\;&\lesssim\;r^{-B}\!\int_0^{r}\rho^B\|g(\rho)\|_{\mathbb{C}^2}\,\ud \rho\;\leqslant\;\int_0^{r}\|g(\rho)\|_{\mathbb{C}^2}\,\ud \rho \\
  &\;\leqslant\;r^{1/2}\,\|g\|_{L^2([0,r],\mathbb{C}^2)}\;=\;o(r^{1/2})
 \end{split}
\]
and 
\[
 \begin{split}
  |b_0^{(g)}(r)\,v_\infty(r)|\;&\lesssim\;r^{B}\!\int_0^{r}\rho^{-B}\|g(\rho)\|_{\mathbb{C}^2}\,\ud \rho\;\leqslant\;r^B\|\rho^{-B}\|_{L^2[0,r]}\|g\|_{L^2([0,r],\mathbb{C}^2)} \\
  &\;\lesssim\;r^{1/2}\,\|g\|_{L^2([0,r],\mathbb{C}^2)}\;=\;o(r^{1/2})\,,
 \end{split}
\]
and \eqref{eq:bvbg_asympt} then follows.
\end{proof}

The next preparatory step is to introduce, for later convenience, the Wronskian of any two square-integrable functions,
\begin{equation}\label{eq:def_wronskian}
 \mathbb{R}^+\ni r\mapsto W_r(\psi,\phi)\;:=\;\det
 \begin{pmatrix}
  \psi^{+}(r) & \phi^+(r) \\
  \psi^{-}(r) & \phi^-(r)
 \end{pmatrix},\quad\psi,\phi\in L^2(\mathbb{R}^+,\mathbb{C}^2)\,,
\end{equation}
and the boundary form for any two functions in $\mathcal{D}(S^*)$,
\begin{equation}\label{eq:boundary_form}
 \omega(g,h)\;:=\;\langle S^*g,h\rangle-\langle g,S^*h\rangle\,,\qquad g,h\in\mathcal{D}(S^*)\,.
\end{equation}
The boundary form is antisymmetric, i.e.,
\begin{equation}\label{eq:omega_antisymm}
 \omega(h,g)\;=\;-\overline{\omega(h,g)}\,,
\end{equation}
and it is related to the Wronskian by
\begin{equation}\label{eq:omegaWr}
 \omega(g,h)\;=\;-\lim_{r\downarrow 0}W_r(\overline{g},h)\,.
\end{equation}
Indeed, using $\widetilde{S}=\mathbf{E}\,{\textstyle\frac{\ud}{\ud r}}+\mathbf{V}(r)$ from \eqref{eq:normal_form}-\eqref{eq:normal_form_potential}, one has
\[
\begin{split}
 \omega(g,h)\;&=\;\int_0^{+\infty}\!\!\ud r\,\big( \langle (\widetilde{S}g)(r),h(r)\rangle_{\mathbb{C}^2}-\langle g(r),(\widetilde{S}h)(r) \rangle_{\mathbb{C}^2} \big) \\
 &=\;\int_0^{+\infty}\!\!\ud r\,\big( \langle \mathbf{E}g'(r),h(r)\rangle_{\mathbb{C}^2}-\langle g(r),\mathbf{E} h'(r) \rangle_{\mathbb{C}^2} \big) \\
 &=\;\lim_{r\downarrow 0}\,\big(\,\overline{g^-(r)}\,h^+(r)-\overline{g^+(r)}\,h^-(r)\big)\;=\;-\lim_{r\downarrow 0}W_r(\overline{g},h)\,.
\end{split}
\]

It is also convenient to introduce the (two-dimensional) space of solutions to the differential problem $\widetilde{S}v=0$,
\begin{equation}
 \mathcal{L}\;:=\;\{v:\mathbb{R}^+\to\mathbb{C}^2\,|\,\widetilde{S}v=0\}\;=\;\mathrm{span}\{v_0,v_\infty\}\,,
\end{equation}
As well known, $r\mapsto W_r(u,v)$ is constant whenever $u,v\in\mathcal{L}$, and this constant is zero if and only if $u$ and $v$ are linearly dependent. It will be important also to keep into account that any $v\in\mathcal{L}$ is square-integrable around $r=0$, as determined in \eqref{eq:asymptotics_for_v0vinf}.

\begin{lem}\label{lem:Lv_vanishes_DS}
For given $v\in\mathcal{L}$, 
\begin{equation}\label{eq:def_Lv}
 \begin{split}
  & L_v:\mathcal{D}(S^*)\to\mathbb{C} \\
  & \;\;\quad\qquad g\longmapsto L_v(g)\;:=\;\lim_{r\downarrow 0}W_r(\overline{v},g)
 \end{split}
\end{equation}
defines a linear functional on $\mathcal{D}(S^*)$ which vanishes on $\mathcal{D}(\overline{S})$.
\end{lem}

\begin{proof}
The linearity of $L_v$ is obvious, and the finiteness of $L_v(g)$ for $g\in\mathcal{D}(S^*)$ is checked as follows. Let us decompose $g=a_0^{(g)}\,v_0+a_\infty^{(g)}\,v_\infty+b_\infty^{(g)}\,v_0+ b_0^{(g)}\,v_\infty$ as in \eqref{eq:decomp_a0ainf_b0binf} and $v=c_0v_0+c_\infty v_\infty$ in the basis of $\mathcal{L}$. Owing to \eqref{eq:def_Lv}, it suffices to control the finiteness of $L_{v_0}(g)$ and $L_{v_\infty}(g)$. By linearity,
\[
 L_{v_0}(g)\;=\;a_0^{(g)}\,L_{v_0}(v_0)+a_\infty^{(g)}\,L_{v_0}(v_\infty)+L_{v_0}(b_\infty^{(g)}\,v_0+ b_0^{(g)}\,v_\infty)\,;
\]
moreover, $L_{v_0}(v_0)=\lim_{r\downarrow 0}W_r(v_0,v_0)=0$, $L_{v_0}(v_\infty)=W_0^{\infty}$, $L_{v_0}(b_\infty^{(g)}\,v_0)=\lim_{r\downarrow 0}\,W_r(v_0,b_\infty^{(g)}\,v_0)=\lim_{r\downarrow 0}\,b_\infty^{(g)}(r)W_r(v_0,v_0)=0$, and $L_{v_0}(b_0^{(g)}\,v_\infty)=\lim_{r\downarrow 0}\,W_r(v_0,b_0^{(g)}\,v_\infty)=\lim_{r\downarrow 0}\,b_0^{(g)}(r)W_r(v_0,v_\infty)=0$. The conclusion is $L_{v_0}(g)=a_\infty^{(g)}W_0^{\infty}$. Analogously, $L_{v_\infty}(g)=-a_0^{(g)}W_0^{\infty}$, and this establishes the finiteness of $L_v(g)$.
Let us now prove now that if $f\in\mathcal{D}(\overline{S})$, then $L_v(f)=0$. Let $\chi\in C^\infty_0([0,+\infty))$ be such that $\chi(r)=1$ for $r\in[0,\frac{1}{2}]$ and $\chi(r)=0$ for $r\in[1,+\infty)$. One has that $v\chi\in\mathcal{D}(S^*)$, indeed $v\chi\in L^2(\mathbb{R}^+,\mathbb{C}^2)$ and
\[
 \begin{split}
  \widetilde{S}(v\chi)\;&=\;(\mathbf{E}\,{\textstyle\frac{\ud}{\ud r}}+\mathbf{V}(r))v\chi\;=\;\chi(\mathbf{E}\,{\textstyle\frac{\ud}{\ud r}}+\mathbf{V}(r))v+\mathbf{E}v\chi' \\
  &=\;(\widetilde{S}v)\chi+\mathbf{E}v\chi'\;=\;\mathbf{E}v\chi'\;\in\;L^2(\mathbb{R}^+,\mathbb{C}^2)\,,
 \end{split}
\]
where we used $\widetilde{S}=\mathbf{E}\,{\textstyle\frac{\ud}{\ud r}}+\mathbf{V}(r)$ and $\widetilde{S}v=0$. Moreover, because of the behaviour of $\chi$ around $r=0$, the Wronskians $W_r(\overline{v\chi},g)$ and $W_r(\overline{v},g)$ are asymptotically equal as $r\downarrow 0$, that is, $ L_{v\chi}=L_{v}$
As a consequence of this latter fact and  of \eqref{eq:omegaWr},
\[
 \begin{split}
  L_{v}(f)\;&=\;L_{v\chi}(f)\;=\;\lim_{r\downarrow 0}W_r(\overline{v\chi},f)\;=\;-\omega(v\chi,f) \\
  &=\;\langle v\chi,S^*f\rangle-\langle S^*(v\chi),f\rangle\;=\;\langle v\chi,\overline{S}f\rangle-\langle v\chi,\overline{S}f\rangle\;=\;0\,,
 \end{split}
\]
which completes the proof.
\end{proof}

We come now to the characterisation of the space $\mathcal{D}(\overline{S})$ which constitutes the main result of this Section.

\begin{prop}\label{prop:characterisationDS}
 Let $f\in\mathcal{D}(S^*)$. The following conditions are equivalent:
 \begin{itemize}
  \item[(i)] $f\in\mathcal{D}(\overline{S})$.
  \item[(ii)] $\omega(f,g)=0$ for all $g\in\mathcal{D}(S^*)$.
  \item[(iii)] $L_v(f)=0$ for all $v\in\mathcal{L}$.
  \item[(iv)] With respect to the decomposition \eqref{eq:decomp_a0ainf_b0binf} for $f$, $a_0^{(f)}=a_\infty^{(f)}=0$.
 \end{itemize}
\end{prop}

\begin{proof} The implication (i)$\Rightarrow$(ii) follows at once from
\[
 \omega(f,g)\;=\;\langle S^*f,g\rangle-\langle f,S^*g\rangle\;=\;\langle\overline{S}f,g\rangle-\langle\overline{S}f,g\rangle\;=\;0\,.
\]
For the converse implication (ii)$\Rightarrow$(i), we observe that
\[
 0\;=\;\omega(f,g)\;=\;\langle S^*f,g\rangle-\langle f,S^*g\rangle\qquad\forall g\in\mathcal{D}(S^*)
\]
is equivalent to $\langle S^*f,g\rangle=\langle f,S^*g\rangle$ $\forall g\in\mathcal{D}(S^*)$, which implies that $f\in\mathcal{D}(S^{**})=\mathcal{D}(\overline{S})$.

The implication (i)$\Rightarrow$(iii) is given by Lemma \ref{lem:Lv_vanishes_DS}. Conversely, let us assume that $L_v(f)=0$ for all $v\in\mathcal{L}$, and let us prove that for such $f$ one has $\omega(f,g)=0$ for all $g\in\mathcal{D}(S^*)$. Since we already established the equivalence (i)$\Leftrightarrow$(ii), we would then conclude that $f\in\mathcal{D}(\overline{S})$, and hence (iii)$\Rightarrow$(i). Owing to the decomposition \eqref{eq:decomp_a0ainf_b0binf} for $g$,
\[
 \omega(f,g)\;=\;a_0^{(g)}\omega(f,v_0)+a_\infty^{(g)}\omega(f,v_\infty)+\omega(f,b_\infty^{(g)}\,v_0)+\omega(f,b_0^{(g)}\,v_\infty)\,.
\]
One has
\[
 \overline{\omega(f,v_0)}\;=\;-\omega(v_0,f)\;=\;\lim_{r\downarrow 0}W_r(\overline{v_0},f)\;=\;L_{v_0}(f)\;=\;0\,,
\]
having used \eqref{eq:omega_antisymm} in the first step, \eqref{eq:omegaWr} in the second, \eqref{eq:def_Lv} in the third, and the assumption $L_v(f)=0$ for all $v\in\mathcal{L}$ in the last step. Analogously,
\[
 \overline{\omega(f,v_\infty)}\;=\;-\omega(v_\infty,f)\;=\;\lim_{r\downarrow 0}W_r(\overline{v_\infty},f)\;=\;L_{v_\infty}(f)\;=\;0\,.
\]
Therefore, $\omega(f,v_0)=\omega(f,v_\infty)=0$, and one is left with
\[
\begin{split}
 \overline{\omega(f,g)}\;&=\;\overline{\omega(f,b_\infty^{(g)}\,v_0)}+\overline{\omega(f,b_0^{(g)}\,v_\infty)} \;=\;-\omega(b_\infty^{(g)}\,v_0,f)-\omega(b_0^{(g)}\,v_\infty,f) \\
 &=\;\lim_{r\downarrow 0}\Big( W_r\big(\,\overline{b_\infty^{(g)}\,v_0}\,,f\big)+W_r\big(\overline{b_0^{(g)}\,v_\infty}\,,f\big)\Big) \\
 &=\;\lim_{r\downarrow 0}\Big( \,\overline{b_\infty^{(g)}(r)}\,W_r(\overline{v_0},f)+\overline{b_0^{(g)}(r)}\,W_r(\overline{v_\infty},f)\Big)\,.
\end{split}
\]
As $r\downarrow 0$, $W_r(\overline{v_0},f)\to L_{v_0}(f)=0$ and $W_r(\overline{v_\infty},f)\to  L_{v_\infty}(f)=0$, and also (as seen in Lemma \ref{lem:decomp_a0ainf_b0binf}) $b_\infty^{(g)}(r)\to 0$ and $b_0^{(g)}(r)\to 0$, whence $\omega(f,g)=0$. This completes the proof of the implication (iii)$\Rightarrow$(i).

Last, in order to establish the equivalence (i)$\Leftrightarrow$(iv), let us decompose $f$ as in \eqref{eq:decomp_a0ainf_b0binf}, namely,
\[
  f\;=\;a_0^{(f)}\,v_0+a_\infty^{(f)}\,v_\infty+b_\infty^{(f)}\,v_0+ b_0^{(f)}\,v_\infty\,,
\]
and let us compute
\[
 \begin{split}
  L_{v_0}(f)\;&=\;\lim_{r\downarrow 0}W_r(\overline{v_0},f)\;=\;a_0^{(f)}\,\lim_{r\downarrow 0}W_r(\overline{v_0},v_0)+a_\infty^{(f)}\,\lim_{r\downarrow 0}W_r(\overline{v_0},v_\infty) \\
  &\qquad\qquad +\lim_{r\downarrow 0}b_\infty^{(f)}(r)\,W_r(\overline{v_0},v_0)+\lim_{r\downarrow 0}b_0^{(f)}(r)\,W_r(\overline{v_0},v_\infty) \\
  &\;=\;a_\infty^{(f)}\,W_0^{\infty}\,.
 \end{split}
\]
Indeed, $W_r(\overline{v_0},v_0)=W_r(v_0,v_0)=0$, and  $W_r(\overline{v_0},v_\infty)=W_r(v_0,v_\infty)\to W_0^{\infty}$, $b_\infty^{(f)}(r)\to 0$, and $b_0^{(f)}(r)\to 0$ as $r\downarrow 0$. Similarly,
\[
 L_{v_\infty}(f)\;=\;-a_0^{(f)}\,W_0^{\infty}\,.
\]
Because of the already proved equivalence (i)$\Leftrightarrow$(iii), we then conclude that $f\in\mathcal{D}(\overline{S})$ if and only if $L_{v_0}(f)=L_{v_\infty}(f)=0$, which from the above computation is tantamount as $a_0^{(f)}=a_\infty^{(f)}=0$. This completes the proof.
\end{proof}

We thus see that Proposition \ref{prop:Sclosure} is therefore an immediate corollary of Lemma \ref{lem:decomp_a0ainf_b0binf} and Proposition \ref{prop:characterisationDS} above.

\begin{proof}[Proof of Proposition \ref{prop:Sclosure}]
The vanishing limit \eqref{eq:f_vanishing} for a generic $f\in\mathcal{D}(\overline{S})$ follows from the fact that, owing to Proposition \ref{prop:characterisationDS}(iv), $f=b_\infty^{(f)}\,v_0+ b_0^{(f)}\,v_\infty$, and from the asymptotics \eqref{eq:bvbg_asympt} of Lemma \ref{lem:decomp_a0ainf_b0binf}. The $H^1$-regularity of $f$ on any interval $[\varepsilon,+\infty)$, with $\varepsilon>0$, follows from the fact that on such interval the $r^{-1}$ potential is bounded and hence the closure of $\mathcal{D}(S)$ in the graph norm is in fact the closure of the smooth and compactly supported functions in the $H^1$-norm.
\end{proof}

\section{Resolvents and spectral gap}\label{sec:resolvent}

In this Section we give the details of the derivation of a couple of relevant consequences from the general classification Theorem \ref{thm:VB-representaton-theorem_Tversion}, which concern the invertibility of each member of the family of self-adjoint extensions and the expression of the resolvent. As an application to the Dirac-Coulomb Hamiltonian under consideration, we then prove Theorem \ref{thm:DC-invertibility-resolvent-gap}.

In fact, Theorem \ref{eq:thm_invertibility_and_resolvent} below is standard within the Kre{\u\i}n-Vi\v{s}ik-Birman extension theory for semi-bounded operators (see, e.g., \cite[Section 6]{GMO-KVB2017}): we present for completeness the proof in the more general framework of self-adjoint extensions of a symmetric operator with a distinguished, invertible extension. An analogous argument, from a somewhat different perspective, can be found in \cite[Theorems 13.8, 13.23, and 13.25]{Grubb-DistributionsAndOperators-2009}.

\begin{thm}[Invertibility of extensions and resolvents]\label{eq:thm_invertibility_and_resolvent}
 Let $S$ be a densely defined symmetric operator on a Hilbert space $\mathcal{H}$ which admits a self-adjoint extension $S_D$ that has everywhere defined and bounded inverse on $\cH$.  In terms of the parametrisation \eqref{eq:ST} of Theorem \ref{thm:VB-representaton-theorem_Tversion}, let $S_T$ be a generic self-adjoint extension of $S$ and $P_T:\mathcal{H} \to \mathcal{H}$ be the orthogonal projection onto $\overline{\mathcal{D}(T)}$, where the operator $T$ is the extension parameter.
 \begin{itemize}
  \item[(i)] $S_T$ is invertible on the whole $\mathcal{H}$ if and only if $T$ is invertible on the whole $\overline{\mathcal{D}(T)}$.
  \item[(ii)] When $S_T$ is invertible, and so is $T$, because of (i), one has
\begin{equation}\label{eq:InversionFormula}
S_T^{-1} = S_D^{-1} + P_T T^{-1} P_T\,.
\end{equation}
  \item[(iii)] Assume further that $\dim\ker S^*=1$, i.e., $S$ has deficiency indices (1,1). Let $\widehat{S}$ be a self-adjoint extension of $S$ other than the distinguished extension $S_D$. Let $\Phi \in \ker S^* \setminus \{0\}$ and for each $z \in\rho(\widehat{S}) \cap \mathbb{R}$ set
\begin{equation}\label{eq:6.4}
\Phi(z):= \Phi + z(S_D -z I)^{-1} \Phi\;\in \;\ker(S^*-z\mathbbm{1})\,.
\end{equation}
Then there exists an analytic function $\eta:\rho(\widehat{S}) \cap \mathbb{R} \to \mathbb{R}$ with $\eta(z) \neq 0$, such that
\begin{equation}\label{eq:KreinResolvent}
(\widehat{S} - z I)^{-1} =(S_D-z I)^{-1} + \eta(z) | \Phi(z) \rangle \langle \Phi(z) |\,.
\end{equation}
$\eta(z)$, $\Phi(z)$ and (\ref{eq:KreinResolvent}) admit an analytic continuation to $\rho(S_D) \cap \rho(\widehat{S})$.
 \end{itemize}
\end{thm}

\begin{proof}
 (i) Let us show first that $S_T$ is injective if and only if $T$ is injective. Assume that $S_T$ is injective and pick $v \in \mathcal{D}(T)$ such that $Tv=0$. Then $v$ is an element in $\mathcal{D}(S_T)$, because it is a vector of the form \eqref{eq:ST}, namely $g = f + S_D^{-1}(Tv+w)+v$, with $f=w=0$. Since $S_Tv=0$ by injectivity one concludes that $v=0$. Conversely if $T$ is injective and for some $g=f+S_D^{-1}(Tv+w)+v \in \mathcal{D}(S_T)$ one has $S_T g=0$, then $\bar S f+Tv+w=0$ Since $\bar S f + Tv +w \in \ran \bar S \boxplus \ran T \boxplus (\ker S^* \cap \mathcal{D}(T)^\perp )$, one must have $\bar S f = Tv = w =0$. Owing to the injectivity of $\bar S$ and $T$, $f=v=0$ and hence $g=0$. Next, let us show that $S_T$ is surjective if and only if $T$ is surjective. One has $\ran S_T=\ran \bar S \boxplus \ran T \boxplus (\ker S^* \cap \mathcal{D}(T)^\perp )$ and in fact $\ran \bar S = \overline{\ran S}$. Thus $T$ is surjective if and only if $\ran T \boxplus (\ker S^* \cap \mathcal{D}(T)^\perp)=\overline{\ran T} \oplus (\ker S^* \cap \mathcal{D}(T)^\perp)=\ker S^*$, if and only if $\ran S_T=\overline{\ran S} \oplus \ker S^* = \mathcal{H}$ if and only if $S_T$ is surjective. The proof of (i) is thus completed.
 
 (ii) (\ref{eq:InversionFormula}) is an identity between bounded self-adjoint operators. For a generic $h \in \ran S_T$ one has $h=S_Tg$ for some $g=f+S_D^{-1}(Tv + w) + v = F+v$, where $f \in \mathcal{D}(S_{min})$, $v \in \mathcal{D}(T)$, $w = \ker S^* \cap \mathcal{D}(T)$ (Theorem \ref{thm:VB-representaton-theorem_Tversion}), and hence $F \in \mathcal{D}(S_D)$. Then
\begin{equation*}
\langle h, S_T^{-1} h \rangle = \langle g, S_T g \rangle =\langle F, S_D F \rangle + \langle v, Tv \rangle.
\end{equation*}
On the other hand
\begin{equation*}
\langle F, S_D F \rangle = \langle S_D F, S_D^{-1} S_D F \rangle = \langle S_T g, S_D^{-1} S_T g \rangle = \langle h, S_D^{-1} h \rangle
\end{equation*}
and
\begin{equation*}
\langle v, T v \rangle = \langle Tv, T^{-1} T v \rangle = \langle P_T S_T g, T^{-1} P_T S_T g \rangle = \langle h, P_T T^{-1} P_T h \rangle
\end{equation*}
whence the conclusion $\langle h, S_T^{-1} h \rangle = \langle h, S_D^{-1} h \rangle + \langle h, P_T T^{-1} P_T h \rangle$.

 (iii) Even without assuming for the moment unital deficiency indices, for $z\in\rho(\widehat{S})\cap\rho(S_D)$  let $T(z)$ be the extension parameter, in the sense of KVB parametrisation \eqref{eq:ST} of Theorem \ref{thm:VB-representaton-theorem_Tversion}, of the operator $\widehat{S} - z \mathbbm{1}$ considered as a self-adjoint extension of the densely defined operator $S(z)=S-z\mathbbm{1}$. Correspondingly, let $P(z)$ be the orthogonal projection onto $\overline{\mathcal{D}(T(z))}$. Then
\begin{equation*}\tag{*}
(\widehat{S} -z \mathbbm{1})^{-1}= (S_D -z\mathbbm{1})^{-1} +P(z) \,T(z)^{-1} P(z)\,,
\end{equation*}
 which follows from part (ii), due to the fact that the distinguished extension of $S-z\mathbbm{1}$ is $S_D-z\mathbbm{1}$.
 Assuming now $\dim\ker S^*=1$, one has $\dim\ker(S^*-z\mathbbm{1})=1$, because of the constancy of the deficiency indices. Moreover, $\widehat{S} -z \mathbbm{1}$ is a self-adjoint extension of $S -z \mathbbm{1}$, whose extension parameter $T(z)$, in the sense of KVB parametrisation of Theorem \ref{thm:VB-representaton-theorem_Tversion}, acts as the multiplication by a real number $t(z)$ on the one-dimensional space $\ker(S^*-z\mathbbm{1})$.  The fact that $(S^*-z\mathbbm{1})\Phi(z)=0$ is obvious by construction. Moreover $\Phi(z)\neq 0$ for each admissible $z$: this is obviously true if $z=0$, and if it was not true  for $z\neq 0$, then 
 $z (S_D-z\mathbbm{1})^{-1} \Phi=-\Phi\neq 0$, which would contradict $\mathcal{D}(S_D-z\mathbbm{1})\cap\ker(S^*-z\mathbbm{1})=\{0\}$. Thus, $\Phi(z)$ spans  $\ker(S^*-z\mathbbm{1})$ and $P_T:=\|\Phi(z)\|^{-2}|\Phi(z)\rangle\langle \Phi(z)|:\cH\to\cH$ is the orthogonal projection onto $\ker(S^*-z\mathbbm{1})$. In this case, the resolvent formula (*) above takes precisely the form \eqref{eq:KreinResolvent} where $\eta(z):=\|\Phi(z)\|^{-2} \,t(z)^{-1}$. Being a product of non-zero quantities, $\eta(z)\neq0$. Moreover,  $z\mapsto(\widehat{S}-z\mathbbm{1})^{-1}$ and $z\mapsto(S_D-z\mathbbm{1})^{-1}$ are analytic operator-valued functions on the whole $\rho(S_D)\cap\rho(\widehat{S})$ (because of the analyticity of resolvents) and so is the vector-valued function $z\mapsto \Phi(z)$ (because of the construction \eqref{eq:6.4}). Therefore, taking the expectation of both sides of (*) on $\Phi(z)$ shows at once that 
 $z\mapsto\eta(z)$ is analytic on $\rho(S_D)\cap\rho(\widehat{S})$, and real analytic on $\mathbb{R}\cap\rho(\widehat{S})$.
\end{proof}

\begin{proof}[Proof of Theorem \ref{thm:DC-invertibility-resolvent-gap}]
 Part (i) is an immediate consequence of Theorem \ref{eq:thm_invertibility_and_resolvent}(i), since the KVB-extension parameter in the present case is the multiplication by $\beta$. This is of course consistent with the representation formula \eqref{eq:Sbeta}, which clearly implies that when $\beta=0$ the extension $S_{\beta=0}$ has a kernel. Analogously, part (ii) is an immediate consequence of Theorem \ref{eq:thm_invertibility_and_resolvent}(ii), because the orthogonal projection $P_T$ has in the present case the expression $P_T=\|\Phi\|^{-2}|\Phi\rangle\langle\Phi|$. Concerning part (iii), \eqref{eq:sigmaess} is a consequence of the fact that, as stated in \eqref{eq:Sbeta-1}, the resolvent difference between the $\beta$-extension and the distinguished extension is compact. Moreover, using \eqref{eq:Sbeta-1} we re-write $S_\beta f = Ef$ as
\begin{equation*}
f=E \,S_\beta^{-1} f=E\,\Big(S_{D}^{-1} +\frac{1}{\,\beta\|\Phi\|^2}\:|\Phi \rangle \langle \Phi |\Big)f\,.
\end{equation*}
This equation is surely solved by $f=0$ and, if $E\in(-E(\beta),E(\beta))$, then the operator acting on the r.h.s.~is a contraction. Thus $f=0$ is the only function which satisfies the eigenvalue equation $S_\beta f = Ef$ and therefore there cannot be eigenvalues in such a regime of $E$.
\end{proof}

\section{Concluding remarks}

We would like to end our analysis with some observations on our overall approach also in comparison with the previous literature.

As documented already, the literature concerning the problem of realising self-adjointly the Dirac-Coulomb operator is vast and unfolds uninterrupted over many decades until recent times, across different disciplines such as ODEs, functional inequalities, operator theory, etc. In fact, the Dirac-Coulomb Hamiltonian is known since long not to be uniquely realised for large Coulomb couplings, with a dominant part of the literature devoted to the study of the properties of the distinguished extension in the critical regime. The perspective of the general classification of the extensions is relatively recent \cite{Voronov-Gitman-Tyutin-TMP2007,Hogreve-2013_JPhysA}, and what we found that was missing was a comprehension of the structure of the family of extensions through the Kre{\u\i}n-Vi\v{s}ik-Birman and Grubb scheme, as opposite to the standard von Neumann scheme. 

In the former framework we could establish Theorems \ref{thm:classification_structure}, \ref{thm:classification_bc}, and \ref{thm:DC-invertibility-resolvent-gap} in a form and through steps that, to our taste, in comparison with \cite[Sections 3 and 4]{Voronov-Gitman-Tyutin-TMP2007} and \cite[Sections 3-7]{Hogreve-2013_JPhysA}, let emerge more straightforwardly the overall extension picture and the meaning of the self-adjointness boundary condition as a multiplicative constraint between regular and singular part of the functions in the domain of the extension, the multiplicative constant giving also immediate information on the invertibility property and on the resolvent and spectral gap of the extension.

This is evident comparing our form \eqref{eq:Sbeta_bc} of the (asymptotic) boundary condition of self-adjointness with \cite[Theorem 7.1, Eq.~(56)]{Hogreve-2013_JPhysA}, purely based on von Neumann's extension theory.

It is also worth pointing out that \cite[Eq.~(67)]{Voronov-Gitman-Tyutin-TMP2007} expresses the vanishing rate of elements of what is here the space $\mathcal{D}(\overline{S})$ only as $O(r^{1/2})$, as $r\downarrow 0$, whereas we proved that the correct vanishing rate is $o(r^{1/2})$ (Proposition \ref{prop:Sclosure}). We record that the $o(r^{1/2})$-rate was mentioned, but not substantiated, already in \cite[Section 2]{Burnap-Brysk-Zweifel-NuovoCimento1981}.

As far as our use of techniques from ODE theory is concerned, most of what we did is somewhat standard, but it has to be highlighted that our analysis of the space $\mathcal{D}(\overline{S})$ in Section \ref{sec:closure} is very much inspired to that of the recent work \cite{B-Derezinski-G-AHP2011} on the `twin' problem of the (scalar) homogeneous Schr\"{o}dinger operator $h=-\frac{\ud^2}{\ud r^2}+\nu r^{-2}$ on half-line. Despite the difference of goals with \cite{B-Derezinski-G-AHP2011}, where the Kre{\u\i}n-Vi\v{s}ik-Birman scheme is not exploited, the resemblance of results is not surprising: in \cite[Proposition 4.17]{B-Derezinski-G-AHP2011} the family of self-adjoint realisations of $h$ is qualified to be a collection $(h_{\theta})_{\theta\in[0,2\pi)}$ where $\mathcal{D}(h_\theta)$ is formed by elements that as $r\downarrow 0$ have the form
\[
 f+c(r^{\frac{1}{2}-m}\cos\theta+r^{\frac{1}{2}+m}\sin\theta)
\]
for some $c\in\mathbb{C}$ and some function $f$ with $f(r)\sim r^{-3/2}$ as $r\downarrow 0$, where $m:=\sqrt{\nu+\frac{1}{4}}$\,. In fact, one would say in the present language that also in that case it is possible to identify a distinguished extension in the regime $\nu>-\frac{1}{4}$, the one with $\theta=\frac{\pi}{2}$, which has the property that $\mathcal{D}(h_{\theta=\pi/2})\subset\mathcal{D}[r^{-2}]$. 

\medskip

\subsection*{Acknowledgment}
We are indebted to Naiara Arrizabalaga, Gianfausto Dell'Antonio, Marko Erceg, Diego Noja, and Giulio Ruzza for many instructive and inspirational discussions on the subject of this work.

\bibliographystyle{siam}
\def\cprime{$'$}

\end{document}